\documentclass[12pt,a4paper]{article}
\usepackage[a4paper,margin=1in,footskip=0.25in]{geometry}
\usepackage{amsthm}
\usepackage{amsmath}
\usepackage{amssymb}
\usepackage[dvipdfmx,unicode,bookmarks=true, bookmarksnumbered=true, bookmarksopen = false, colorlinks=true]{hyperref}
\usepackage{graphics}
\usepackage{epsfig}
\usepackage{geometry}
\usepackage{caption}
\usepackage{multirow}
\makeatletter
\newenvironment{subtheorem}[1]{%
  \def\subtheoremcounter{#1}%
  \refstepcounter{#1}%
  \protected@edef\theparentnumber{\csname the#1\endcsname}%
  \setcounter{parentnumber}{\value{#1}}%
  \setcounter{#1}{0}%
  \expandafter\def\csname the#1\endcsname{\theparentnumber.\Alph{#1}}%
  \ignorespaces
}{%
  \setcounter{\subtheoremcounter}{\value{parentnumber}}%
  \ignorespacesafterend
}
\makeatother
\newcounter{parentnumber}

\makeatletter \renewenvironment{proof}[1][\proofname]
{\par\pushQED{\qed}\normalfont\topsep6\p@\@plus6\p@\relax\trivlist\item[\hskip\labelsep\bfseries#1\@addpunct{.}]\ignorespaces}{\popQED\endtrivlist\@endpefalse} \makeatother

\theoremstyle{plain}
\newtheorem{thm}{Theorem}
\newtheorem{lem}{Lemma}
\newtheorem{defn}{Definition}

\newtheorem{cor}{Corollary}
\newtheorem{prop}{Proposition}

\begin{document}
\author{Yangbo Song\thanks{Department of Economics, UCLA. Email: darcy07@ucla.edu.}}
\date{October 8, 2015}
\title{Social Learning with Coordination Motives}
\maketitle

\begin{abstract}

The theoretical study of social learning typically assumes that each agent's action affects only her own payoff. In this paper, I present a model in which agents' actions directly affect the payoffs of other agents. On a discrete time line, there is a community containing a random number of agents in each period. Before each agent needs to take an action, the community receives a private signal about the underlying state of the world and may observe some past actions in previous communities. An agent's payoff is higher if her action matches the state or if more agents take the same action as hers. I analyze two observation structures: exogenous observation and costly strategic observation. In both cases, coordination motives enhance social learning in the sense that agents take the correct action with significantly higher probability when the community size is greater than a threshold. In particular, this probability reaches one (asymptotic learning) with unbounded private beliefs and can be arbitrarily close to one with bounded private beliefs. I then discuss the issue of multiple equilibria and use risk dominance as a criterion for equilibrium selection. I find that in the selected equilibria, the community size has no effect on learning under exogenous observation, facilitates learning under endogenous observation and unbounded private beliefs, and either helps or hinders learning under endogenous observation and bounded private beliefs.

\begin{flushleft}
\end{flushleft}

\begin{flushleft}
\end{flushleft}
\end{abstract}

\newpage
\section{Introduction}

The study of social learning focuses on how valuable information is transmitted in a society of self-interested and strategic agents as well as how dispersed and decentralized information is aggregated to facilitate greater precision of knowledge. A typical situation involves a large number of individuals who make a single decision sequentially. The payoff of this decision depends on an unknown state of the world, about which each individual is given a noisy signal. The state of the world may refer to different economic variables in different applications, such as the quality of a new product, the return on an investment opportunity, or the intrinsic value of a research project. The probabilistic distribution of signals depends on the state and is assumed to be distinctive for each possible value of this state. Hence, if signals were observable, the aggregation of signals would be sufficient for individuals to ultimately learn the value of the state with near certainty. However, because signals are private and often cannot be transmitted via direct communication, an individual must extract information from observations of her predecessors' decisions to determine her own actions. A general and important question thus arises: what behaviors and observation structures can lead to the level of learning achieved by efficient information aggregation? In other words, under what conditions will observation reveal the true state, and how likely is it that agents will make the correct decision?

The above framework has been adopted widely in the literature, including but not exclusive to the notable study of herding behavior and information cascades in various applications, such as investments\cite{ScharfsteinandStein}, bank runs\cite{DiamondandDybvig} and technology adoption\cite{Choi}. Among the literature that provides a theoretical analysis, renowned early research by Bikhchandani, Hirshleifer and Welch\cite{BHW}, Banerjee\cite{Banerjee} and Smith and Sorensen\cite{SS} demonstrates that efficient information aggregation may fail: in a perfect Bayesian equilibrium, individuals eventually herd in choosing the wrong action with positive probability. Recent works such as Acemoglu et al.\cite{ADLO} consider a more general observation structure and note that society's learning of the true state depends on two factors: the possibility of arbitrarily strong private signals and the nonexistence of excessively influential individuals.

However, despite the large body of theoretical literature on social learning and information externalities, most models fail to account for a crucial factor that influences individual strategic behaviors: \textit{coordination motives}. In an environment with coordination motives, one agent's action may directly affect another agent's payoff. The absence of this effect greatly limits the range of applications that can be analyzed using the standard framework because coordination motives are prevalent in many strategic environments that involve social learning, ranging from the choice of computer software to the choice of research areas. In addition, the very existence of coordination motives often facilitates local information sharing in both signaling and observations because individuals in such situations have mutual interests in such sharing. Hence, one should expect to see very different patterns of action as well as information updating in the observational learning process. Finally, most existing studies assume that observation is given by some exogenous stochastic process, but in many applications, it is part of an agent's strategic decision. Presumably, once observation becomes a choice, it should have an immediate effect on the accuracy of action and should also change how coordination motives influence social learning. Therefore, a more general framework is needed to include these important elements in the study of social learning and to fully understand their impact.

%First, it is often assumed that observation is costless and non-strategic: observation is given according to some exogenous stochastic process, but agents are not able to choose whose actions to observe. In many applications, however, observation is typically both costly and strategic. On one hand, time and resources are required to obtain information regarding others' actions. On the other hand, an agent would naturally choose to observe the presumably more informative actions based on the positions of individuals in the decision sequence. For instance, in the case of technology diffusion, suppose that a firm would like to find out whether its competitors in a given region have adopted a newly invented technology before it makes its own decision. This information is usually not acquired in a free and random fashion, but requires money and labor to set up an investigation group targeting predetermined competitors. The fact that different observation structures apply to different applications calls for a model that is rich enough to analyze each observation structure independently.

To focus on the typical strategic environment with the above features, consider the following example. There is a group of consumers who need to decide which one of two possible smartphones to switch to in their usage. The sequence of actions is determined by the expiration dates of their current contract. Among this group, there are smaller ``communities'' of consumers (for example, college friends who enroll in the same wireless phone package) that make their decisions within a relatively small period of time. For a consumer in an arbitrary community, because interaction is more convenient among people using the same model, she prefers that others use the model that she chooses. Before she makes her own decision, she may observe some previously made decisions from other communities. Such observations may be exogenously given: they may simply come from noticing which smartphones other people in her social network are using. Alternatively, observations may be strategically chosen: the consumer can pay a registration fee to enter an online forum where she can see other consumers' choices with corresponding time stamps. If she is not able to view all the available posts, her rational choice is to select the most informative posts. Finally, regardless of the observation structure, she will most likely share her observations with others in her community but not with outsiders.

In this paper, I propose a model that is consistent with the framework of Bikhchandani, Hirshleifer and Welch\cite{BHW} and Acemoglu et al.\cite{ADLO} and is simultaneously flexible to analyze coordination motives under different observation structures. More formally, there is an underlying state of the world that is binary in value and cannot be observed directly. On an infinite and discrete time line, there is a community of random size in each period, and its members (the agents) each take a binary action simultaneously. The payoff of an agent depends on whether her action matches the state as well as what actions are taken by others in her community. The more agents take the same action as she does, the higher the payoff she enjoys. At the beginning of the period, the agents in the community obtain a noisy signal about the true state\footnote{The assumption of one signal for each community is used without loss of generality in the case of local sharing. Equivalently, we could assume that each agent has one signal that she shares only with others in her community.}. The value of this signal is common knowledge within the community but cannot be observed by any other community.

After obtaining the signal but before taking the action, the agents simultaneously observe a subset of actions of their predecessors, i.e., agents from previous communities. The observed actions are locally shared information as well: in other words, actions observed by one agent are also observed by every other agent within the same community. Observation is \textit{exogenous} if it is pre-determined, regardless of the signal and the community size. Observation is \textit{endogenous} if each agent can choose to pay a fixed cost and select a given number of ordered actions to observe.

Hence, there are three central determinants of the pattern of social learning: signal strength (bounded or unbounded beliefs), observation structure (exogenous or endogenous) and strength of coordination motives (community size). This paper establishes the first theoretical framework to understand the interaction among these factors and to answer the question of when observation is truth-revealing and whether asymptotic learning occurs in this more realistic but complex environment. In particular, I highlight the contrast between the pattern of learning in a singleton community and that in a potentially large community. As can be expected, coordination motives in a large community bring about an additional ``desire to conform'' that is absent when each agent cares only about her own action. However, as suggested by my major findings summarized below, the incentive to decide based on the group is not entirely negative. On the contrary, coordination motives may improve learning in each combination of signal structure and observation structure.

First, suppose that observation is exogenous. When beliefs are unbounded, meaning that a private signal may be arbitrarily informative about the true state, agents can almost surely know the state from their observed action sequence if they always observe an action that has been taken recently. Moreover, because observation is independent from signal value, the stronger notion \textit{asymptotic learning} can be achieved: not only do agents know the true state with near certainty, but their actions converge to the ``correct'' action as well. This result holds regardless of the community size and is consistent with existing results in the theoretical literature.

When beliefs are bounded, coordination motives facilitate better learning. Previous research has shown that when there is only one agent in each period (i.e., when the community size is one), observation never reveals the truth over time. However, I show that when the community size becomes sufficiently large, there exists an equilibrium with truth-telling observation. In such an equilibrium, an agent may take either of the two actions for any possible posterior belief that she has on the true state: given a certain range of a signal, she takes her action according to the signal and otherwise acts according to the observation. A rough intuition for this result is that when the community size is sufficiently large, if all but one agent in a community choose one action, it would be optimal for the remaining agent to choose the same action even if it is unlikely to be the ``correct'' action. Hence, even under bounded beliefs, it is possible for all agents to base their actions on their signal for a non-zero measure of signals, regardless of when they each move in the action sequence. The signal effectively serves as a correlation device, which ensures efficient information aggregation from observing the actions of predecessors. As a further result, depending on the construction of equilibrium strategies, the probability of taking the correct action can become arbitrarily close to $1$ at the limit; hence, asymptotic learning can be approximated under strong coordination motives even if the most precise signal has only limited information value. Indeed, as long as the probability of agents acting only according to the signal is positive, observation reveals the truth at the limit. Hence, decreasing this probability in turn increases the probability of taking the correct action.

Now suppose that observation is endogenous. An initial observation is that even under unbounded private beliefs, asymptotic learning is not achievable: the probability of taking the correct action is always bounded away from $1$. The reason is that with costly observation, an agent is not willing to observe whenever her signal is sufficiently precise but still not perfect. I then give a sufficient and necessary condition for truth-telling observation: the size of an agent's observed neighborhood becomes arbitrarily large over time. Because of the impossibility of asymptotic learning, any observation of finitely many actions has erroneous implication on the true state with positive probability; this probability of error can be eliminated once infinitely many actions are observed.

In the presence of coordination motives, the equilibrium learning patterns change significantly. When the community size is potentially large, an equilibrium emerges with asymptotic learning (with unbounded beliefs) or approximate asymptotic learning (with bounded beliefs). Such improvement in learning is driven by the possibility of incentivizing observation in a large community. For example, imagine a community in which all but one agent choose \textit{not} to observe any action. For the remaining agent, her observation is very valuable to both her peers and herself because when agents care about the actions of other agents, even a small improvement in learning about the true state brings a considerable increase in all of their payoffs. Moreover, additional incentives for observation can be provided by the credible threat of conforming to a sub-optimal action in the absence of observation. Following this intuition, I show that there exists an equilibrium in which at least one agent always chooses to observe for any value of the private signal. Therefore, the argument under exogenous observation can be applied to establish or approximate asymptotic learning. This implies that the negative incentive for observation, as induced by observation cost, can be eliminated by the marginal benefit of observation under coordination motives. At the same time, efficient information aggregation still exists as a result of either unbounded beliefs or a small but positive probability of coordinated actions based on signal only.

One prominent difference arising from the inclusion of coordination motives in the model is that multiple equilibria arise in general, in contrast to the generically unique equilibrium with singleton communities. In the discussion section, I address the issue of equilibrium selection by imposing the criterion of risk dominance. I show that the equilibrium in which each agent always maximizes the probability of action matching the state is risk dominant, and I reveal that in this equilibrium, stronger coordination motives still lead to better learning. Under bounded beliefs, however, the risk-dominant equilibrium has different implications: depending on the observation structure, the equilibrium learning probability with coordination motives may be higher, lower or unchanged relative to that with singleton agents.

The remainder of this paper is organized as follows. Section 2 provides a review of the related literature. Section 3 introduces the model. Sections 4 and 5 respectively present the main results under exogenous observation and endogenous observation. Section 6 discusses some additional features and extensions of the model. Section 7 concludes the paper. All proofs are included in the Appendix.

\section{Literature Review}

A large and growing body of literature examines the problem of social learning by Bayesian agents who can observe others' choices. This literature begins with Bikhchandani, Hirshleifer and Welch\cite{BHW} and Banerjee\cite{Banerjee}, who first formalize the problem systematically and concisely and identify information cascades as the cause of herding behavior. In their models, the informativeness of the observed action history outweighs that of any private signal with a positive probability, and herding occurs as a result. Smith and Sorensen\cite{SS} propose a comprehensive model of a similar environment with a more general signal structure. They show that signal strength plays a decisive role in social learning in the sense that the possibility of arbitrarily strong signals is necessary and sufficient for asymptotic learning in their framework. The concepts of bounded and unbounded private beliefs introduced by those authors will play an important role in the remainder of the current paper. These seminal papers, along with the general discussion by Bikhchandani, Hirshleifer and Welch\cite{BHW2}, assume that agents can observe the entire previous decision history, i.e., the whole ordered set of choices of their predecessors. This assumption can be regarded as an extreme case of exogenous observation structure. Related contributions to the literature include Lee\cite{Lee}, Banerjee\cite{Banerjee2} and Celen and Kariv\cite{CK}, where agents observe only a given fraction of the entire decision history.

A more recent paper by Acemoglu et al.\cite{ADLO} studies the environment in which each agent receives a private signal about the underlying state of the world and observes (some of) their predecessors' actions according to a general stochastic process of observation. Their main result states that when the private signal structure features unbounded belief, asymptotic learning occurs in each equilibrium if and only if the observation structure enables agents to always observe some close predecessor. Other recent research in this area include the works of Banerjee and Fudenberg\cite{BF}, Gale and Kariv\cite{GK}, Callander and Horner\cite{CH} and Smith and Sorensen\cite{SS2}, which differ from Acemoglu et al.\cite{ADLO} mainly in making alternative assumptions regarding observation, e.g., agents observe only the number of other agents taking each available action but not the positions of the observed agents in the decision sequence.

Two common assumptions made in the abovementioned literature are exogenous observation and pure informational externalities; according to the latter, an agent cares only about taking the correct action, and her payoff is not directly affected by others' actions. The literature exploring the relaxation of either of these assumptions is relatively under-developed. A few recent papers initiated the discussion on the impact of costly observations on social learning. In Kultti and Miettinen\cite{KM2}\cite{KM1}, both the underlying state and the private signal are binary, and an agent pays a cost for each action that she observes. In Celen\cite{Celen}, the signal structure is similar to the general one adopted in this paper, but it is assumed that an agent can pay a cost to observe the entire history of actions before hers. A much richer model is given by Song\cite{Song}, as it allows for the most general signal structure as well as the possibility that agents would need to strategically choose a proper subset of their predecessors' actions to observe. A major implication from these works is that the existence of observation costs prevents asymptotic learning, although it may increase the informativeness of an observed action sequence because agents will sometimes rationally choose not to observe and rely on their signal.

The theoretical literature on the interplay between information cascades and coordination motives is also rather small. Moreover, the few existing papers often differ from one another in important aspects, such as the payoff function, the sequence of moves and the information update process (see, e.g., Choi\cite{Choi}, Dasgupta\cite{Dasgupta}, Jeitschko and Taylor\cite{JT}, Frisell\cite{Frisell}, Vergari\cite{Vergari}). However, there is also a small group of experimental studies of information cascades and payoff externalities (see, e.g., Hung and Plott\cite{HP}, Drehmann et al.\cite{Drehmann}). The major results of those studies suggest a learning pattern that is consistent with this paper: when agents care about the actions of one another beyond the informational externalities, they are both more likely to conform and more likely to take the ``correct'' action. Informational herding is thus reduced.

This paper can be positioned in line with the works of Bikhchandani, Hirshleifer and Welch\cite{BHW}, Smith and Sorensen\cite{SS}, Acemoglu et al.\cite{ADLO} and others in the sense that I adopt the general signal structure and the sequential decision process developed in these models. Nevertheless, this paper differs from the previous research in two important aspects. First, instead of assuming an exogenous observation structure, I allow observation to occur as part of an agent's strategic decision. Second, in addition to informational externalities, my model also features payoff externalities: the more agents take the same action, the higher the payoff each agent enjoys. As shown subsequently in this paper, these assumptions not only are more realistic in most applications but also have significant impacts on the equilibrium learning pattern.

In this paper and in most of the cited theoretical papers above, agents are assumed to update their beliefs according to Bayes' rule. There is also a well-known body of literature on non-Bayesian observational learning. In these models, rather than applying Bayes' update to obtain the posterior belief regarding the underlying state of the world by using all the available information, agents may adopt some intuitive rule of thumb to guide their choices (Ellison and Fudenberg\cite{EF}\cite{EF2}), may update their beliefs according to only part of their information (Bala and Goyal\cite{BG}\cite{BG2}), may naively update their beliefs by taking weighted averages of their neighbors' beliefs (Golub and Jackson\cite{GJ}), or may be subject to bias in interpreting information (DeMarzo, Vayanos and Zwiebel\cite{DVZ}).

Finally, the importance of observational learning has been well documented in both empirical and experimental studies, in addition to those already mentioned. Both focusing on the adoption of new agricultural technology, Conley and Udry\cite{CU} and Munshi\cite{Munshi2} not only support the importance of observational learning but also indicate that observation is often constrained because, in practice, a farmer may be unable to receive information regarding the choice of every other farmer in the area. Munshi\cite{Munshi} and Ioannides and Loury\cite{IL} demonstrate that social networks play an important role in individuals' acquisition of employment information. Cai, Chen and Fang\cite{CCF} conduct a natural field experiment to indicate the empirical significance of observational learning in which consumers obtain information about product quality from the purchasing decisions of others.

\section{Model}

\subsection{Private Signal Structure}

Consider a discrete and infinite time line: $t=1,2,...$. At each period $t$, there is a set of agents $N^t$ that move simultaneously. We refer to $N^t$ as a \textit{community}. The community size $Q^t=|N^t|$ is randomly selected from a commonly known probability distribution $G$ on $\mathbb{N}^+$, with the largest community size in the support being finite. The $Q^t$ values are independent and identically distributed (i.i.d.) over time. $Q^t$ is common knowledge for agents in $N^{t'}$, $t'\geq t$, at the beginning of period $t'$.

Let $\theta\in\{0,1\}$ be the state of the world with equal prior probabilities, i.e., $Prob(\theta=0)=Prob(\theta=1)=\frac{1}{2}$. Given $\theta$, an i.i.d. private signal $s^t(Q^t)\in S=(-1,1)$ is realized in period $t$ after the realization of $Q^t$, which is observed by every agent in $N^t$ and by no one else. The $s^t(Q^t)$ values are independently distributed, but their distributions can be heterogeneous depending on $Q^t$. One interpretation of this setting is that each agent receives and shares a signal with the community; the more agents there are, the more precise the aggregated information about the true state is.

The probability distributions regarding the signal conditional on the state are denoted as $F^{Q}_0(s)$ and $F^{Q}_1(s)$ (with continuous density functions $f^{Q}_0(s)$ and $f^{Q}_1(s)$), where $Q$ denotes the community size. The pair of measures $\{F^{Q}_0,F^{Q}_1\}_{Q\in\mathbb{N}^+}$ are referred to as the \textit{signal structure}, and I assume that the signal structure has the following properties for every $Q$:
\begin{itemize}
\item{1.} The pdf's $f^Q_0(s)$ and $f^Q_1(s)$ are continuous and non-zero everywhere on the support, which immediately implies that no signal is fully revealing of the underlying state.
\item{2.} Monotone likelihood ratio property (MLRP): $\frac{f^Q_1(s)}{f^Q_0(s)}$ is strictly increasing in $s$. This assumption is made without loss of generality: as long as no two signals generate the same likelihood ratio, the signals can always be re-aligned to form a structure that satisfies the MLRP.
\end{itemize}

The focus of this paper is to examine the interaction among signal, observation and externalities and to identify conditions that need to be imposed on each factor to ensure the highest possible level of learning. To address this issue and present the major findings, it is useful to first introduce a notation that categorizes the signal structure. The \textit{private belief} of an agent is defined by the probability of the true state being $1$ according to her signal only, and it is given by $\frac{f^Q_1(s)}{f^Q_0(s)+f^Q_1(s)}$.

\begin{defn}
We say that agents have \textit{unbounded private beliefs} if $\lim_{s\rightarrow 1}\frac{f^Q_1(s)}{f^Q_0(s)+f^Q_1(s)}=1$ and $\lim_{s\rightarrow -1}\frac{f^Q_1(s)}{f^Q_0(s)+f^Q_1(s)}=0$ for some $Q$ on the support of distribution $G$. We say that agents have \textit{bounded private beliefs} if  $\lim_{s\rightarrow 1}\frac{f^Q_1(s)}{f^Q_0(s)+f^Q_1(s)}<1$ and $\lim_{s\rightarrow -1}\frac{f^Q_1(s)}{f^Q_0(s)+f^Q_1(s)}>0$ for every $Q$ on the support of distribution $G$.
\end{defn}

Unbounded private beliefs correspond to a situation in which a community may receive an arbitrarily strong signal about the underlying state, while bounded beliefs indicate that the amount of information that can be derived from a single private signal is always limited.

\subsection{The Sequential Decision Process}

The agents in $N^t$ each take a single action simultaneously between $0$ and $1$. Let $a_n^t\in\{0,1\}$ denote the action of agent $n$ in $N^t$.

Agent $n$ cares about the action of every agent in $N^t$. Given $\{a_i^t:i\in N^t\}$, the payoff of agent $n$ is equal to $u(\theta,a_n^t,m)>0$, where $m$ is the number of actions in $\{a_i^t:i\in N^t\}$ that are the same as $a_n^t$. I make the following assumptions about $u$:
\begin{itemize}
\item{1.} Given every $\theta,a_n^t$, $u$ is increasing in $m$. This assumption means that every agent prefers that more of her peers take the same action as she does.
\item{2.} Given every $m$, $u(0,0,m)=u(1,1,m)>u(1,0,m)=u(0,1,m)$. This assumption means that every agent prefers to take the ``correct'' action, i.e., the action that matches the state.
\item{3.} When $m$ is sufficiently large, $u(a,b,m)>u(a,1-b,1)$, $a,b\in\{0,1\}$. This assumption means that coordination motives (conforming to the majority) can dominate information motives (matching the state) in a large community.
\end{itemize}

The direct influence of every agent's action on the payoffs of other agents within the same community differentiates this model from most theoretical literature on social learning. In addition to the widely studied informational externalities that arise from sequential observation, there now exists a new parallel economic force, coordination motives, that generates an incentive for an agent to conform with her peers. This incentive becomes stronger as the community size increases. The primary goal of this paper is to ascertain how this incentive affects individual behavior as well as the overall learning level and to determine whether it improves or impairs the likelihood of agents taking the correct action over time.

After receiving signal $s^t(Q^t)$ and before engaging in the above action, the agents may observe some of the actions taken by their predecessors. In this paper, I will discuss two possible structures of observation.

\subsubsection{Exogenous observation}

The agents in $N^t$ observe the ordered action sequence in a \textit{neighborhood} $B^t\subset\{\cup_{N^i\in\mathcal{N}}N^i:\mathcal{N}\subset\{N^1,\cdots,N^{t-1}\}\}$ (each agent in $N^t$ observes the same action sequence). The neighborhood $B^t$ is generated according to a probability distribution $O^t$ over the set $\{N^1,\cdots,N^{t-1}\}$. The draws from each $O^t$ are independent from one another for all $t$ and from the realization of community size and private signals. Let $\bar{B}^t=\cup_{B^t:O^t(B^t)>0}B^t$ be the union of all possible neighborhoods that can be observed in period $t$. The sequence $\{O^t\}_{t\in\mathbb{N}^+}$ is called the \textit{observation structure} and is common knowledge, while the realization of $s^t$ and $B^t$ are known by agents only in $N^t$.

%The neighborhood $B^t$ is generated according to an arbitrary probability distribution $G^t$ over the set of all possible unions of subsets of $\{N^1,\cdots,N^{t-1}\}$. Let $\bar{B}^t=\cup_{B^t:H^t(B^t)>0}B^t$ be the union of all possible neighborhood that can be observed in period $t$. The draws from each $G^t$ are independent from each other for all $t$ and from the realization of community size and private signals.

Let $H^t=\{a_m\in\{0,1\}:m\in B^t, B^t\subset \bar{B}^t\}$ denote the set of actions that $n$ can possibly know from observation, and let $h^t$ be a particular action sequence in $H^t$. Let $I^t=\{s^t(Q^t),h^t\}$ be $n$'s \textit{information set}. Note that the information set of every agent in $N^t$ is the same. The set of all possible information sets of $n$ is denoted as $\mathcal{I}^t$.

A \textit{strategy} for $n$ is a mapping $\phi_n^t: \mathcal{I}^t\rightarrow \{0,1\}$ that selects a decision for every possible information set. A \textit{strategy profile} is a sequence of strategies $\phi=\{\phi^t\}_{t\in\mathbb{N}^+}=\{\{\phi_n^t\}_{n\in\{1,\cdots,Q^t\}}\}_{t\in\mathbb{N}^+}$. I use $\phi_{-n}^t=\{\phi_{n'}^t\}_{n'\neq n}$ to denote the strategies of all agents other than $n$ in period $t$, $\phi_{-t}=\{\phi^{t'}\}_{t'\neq t}$ to denote the strategies of all agents other than those in $N^t$, and $\phi_{-n,t}=(\phi_{-n}^t,\{\phi^{t'}\}_{t'\neq t})$ to denote the strategies of all agents other than $n$.

Given a strategy profile, the sequence of decisions $\{a_n^t\}_{n\in\mathbb{N}}$ is a stochastic process. I denote the probability measure generated by this stochastic process as $\mathcal{P}_{\phi}$.

\subsubsection{Endogenous Observation}

The agents in $N^t$ simultaneously acquire information about the previous decisions of other agents through observation. Each agent $n$ can pay a cost $c>0$ to obtain a \textit{capacity} $K(t)\in\mathbb{N}^+$; otherwise, he pays nothing and chooses $\varnothing$.

With capacity $K(t)$, agent $n$ can select a \textit{neighborhood} $B^t(n)\subset\cup_{i=1}^{t-1}N^i$ of maximum size $K(t)$, i.e., $|B^t(n)|\leq K(t)$, and observe the action of each agent in $B^t(n)$. The actions in $B^t(n)$ are observed simultaneously, and no agent can choose any additional observation based on what she has already observed. Let $\mathcal{B}^t(n)$ denote the set of all possible neighborhoods that $n$ can observe. After the agents make their decisions regarding observation, the actions that they choose to observe are revealed and become public information within $N^t$. That is, every agent in $N^t$ observes $B^t=\cup_{n=1}^{Q^t} B^t(n)$.

An agent's strategy in the above sequential game consists of two problems: (1) given her private signal, whether to make costly observation and, if so, whom to observe; (2) after observation (or not), which action to take between $0$ and $1$ given the realization of observed actions. With some abuse of notation, let $H^t=\{a_m\in\{0,1\}:m\in B\subset \cup_{i=1}^{t-1}Q^i,|B|\leq Q^tK(t)\}$ denote the set of actions that $n$ can possibly know from observation by herself and others, and let $h^t$ be a particular action sequence in $H^t$. $I^t=\{s^t(Q^t),h^t\}$ and $\mathcal{I}^t$ are defined similarly as that used above.

A \textit{strategy} for $n$ is the set of two mappings $\sigma_n^t=(\sigma_n^{t,1},\sigma_n^{t,2})$, where $\sigma_n^{t,1}:S\rightarrow \mathcal{B}^t(n)$ selects $n$'s choice of observation for every possible private signal and $\sigma_n^{t,2}:\mathcal{I}^t\rightarrow \{0,1\}$ selects an action for every possible information set. A \textit{strategy profile} is a sequence of strategies $\sigma=\{\sigma^t\}_{t\in\mathbb{N}^+}=\{\{\sigma_n^t\}_{n\in\{1,\cdots,Q^t\}}\}_{t\in\mathbb{N}^+}$. I use the notation $\sigma_{-n}^t=\{\sigma_{n'}^t\}_{n'\neq n}$, $\sigma_{-t}=\{\sigma^{t'}\}_{t'\neq t}$ and $\sigma_{-n,t}=(\sigma_{-n}^t,\{\sigma^{t'}\}_{t'\neq t})$ in a manner similar to that used for exogenous observation.

Given a strategy profile, the sequence of decisions $\{a_n^t\}_{n\in\mathbb{N}}$ is a stochastic process. I denote the probability measure generated by this stochastic process as $\mathcal{P}_{\sigma}$.

A decisive difference between exogenous and endogenous observation lies in how observation correlates with signal. Under exogenous observation, no correlation between signal and observation exists because they are simply two independent processes. Under endogenous observation, however, observation--whether to observe and, if so, whom to observe--may depend on the value of the private signal because it is now part of an agent's optimal decision. Conceivably, for an agent who attempts to extract information about the true state from her observation, her inference on private signals and the observation of her predecessors, which then partially determines her posterior belief on the state, will be formed very differently under the two observation structures. As shown in subsequent sections of the paper, observation structure has a significant impact on the pattern of social learning.

\subsection{Perfect Bayesian Equilibrium}

\begin{defn}
A strategy profile $\sigma^*$ (resp. $\phi^*$) is a pure strategy \textbf{perfect Bayesian equilibrium} (PBE) if, for each $t\in\mathbb{N}^+$ and $n\in\{1,\cdots,Q^t\}$, $\sigma_n^{t*}$ is such that given $\sigma^*_{-n,t}$, (1) $\sigma_n^{*t,2}(I^t)$ (resp. $\phi_n^{*t}(I^t)$) maximizes the expected payoff of $n$ given every $I^t\in \mathcal{I}^t$ and (2) $\sigma_n^{*t,1}(s_n^t)$ maximizes the expected payoff of $n$, given every $s_n^t$ and given $\sigma_n^{*t,2}$.
\end{defn}

Whether observation is exogenous or endogenous, the idea underlying PBE is similar: given all available information and the strategy of each predecessor and each peer, an agent determines her payoff-maximizing strategy. In a model without coordination motives, this strategy always coincides with the strategy that maximizes the probability of taking the correct action, but in this situation, it may not because the actions of one's peers must also be considered. An equilibrium strategy under endogenous observation differs from one under exogenous observation in its additional component of observation choice after receiving the private signal. In such a case, an agent optimizes her observation according to her signal value and others' strategies.

Throughout the remainder of the paper, I simply refer to PBE as ``equilibrium''.

\begin{prop}
In every equilibrium $\sigma^*$ (resp. $\phi^*$) and for every $t$, actions are always unanimous in $N^t$: for every $I^t$, $\sigma_n^{*t,2}(I^t)=\sigma_m^{*t,2}(I^t)$; for every $m,n\in N^t$.
\end{prop}

Proposition 1 indicates an agent's incentive to conform to her peers in the same community. Note that the posterior belief on the true state is the same across the community, and consider the two sub-groups of agents choosing different actions. If an agent choosing action $1$ weakly prefers $1$ to $0$, then each agent choosing $0$ must strictly prefer $1$ to $0$. This is a contradiction, and hence, the only equilibrium action profile is unanimous. This result initially seems to indicate that coordination motives always exacerbate herding and are harmful for learning because there is now an additional incentive to ignore one's signal and submit to the majority. However, this result also implies that agents in a community may conform to an action profile that depends on their signal rather than on observation, such that their actions become \textit{more} informative for successors. As will be shown later, such behavior indeed improves social learning to a great extent.

Notably, indifference between the two actions can exist in a \textit{mixed strategy} equilibrium. In fact, when the community size is large, there always exists a mixed strategy equilibrium in which an agent's probability of mixing between $1$ and $0$ depends on the signal value. However, because the mixed strategy equilibrium does not provide additional insight into the relation between social learning and coordination motives, I will not discuss it in detail in this paper.

\subsection{Learning}

The main focus of this paper is to determine what type of information aggregation will result from equilibrium behavior. First, I define the different types of learning studied in this paper.

\begin{defn}
An equilibrium $\sigma^*$ (resp. $\phi^*$) has \textbf{asymptotic learning} if every agent takes the correct action at the limit:
\begin{align*}
\lim_{t\rightarrow\infty}\mathcal{P}_{\sigma^*}(a_n^t=\theta)=1\text{ for all $n$}.
\end{align*}
\end{defn}

In this paper, the unconditional probability of taking the correct action, $\mathcal{P}_{\sigma^*}(a_n^t=\theta)$, is also referred to as the \textit{learning probability}. Asymptotic learning requires that this probability converge to $1$, i.e., the posterior beliefs converge to a degenerate distribution on the true state. In terms of information aggregation, asymptotic learning can be interpreted as equivalent to making all private signals public and thus aggregating information efficiently. It marks the upper bound of social learning with any signal structure and observation structure.

Asymptotic learning may not always be achieved, especially under an endogenous observation structure, because a rational agent may choose not to make costly observations when her signal is already quite precise. In such a case, it is still interesting to determine whether information can be efficiently aggregated via observation, i.e., to ask the following question: \textit{when an agent decides to observe}, will her observation reveal the truth and lead her to act correctly? A formal analysis calls for the notion of \textit{truth-telling observation}, which is defined below.

Let $\hat{a}^t$ be a hypothetical action that is equal to the state with higher posterior probability given any $I^t$.

\begin{defn}
An equilibrium $\sigma^*$ (resp. $\phi^*$) has \textbf{truth-telling observation} if $\hat{a}^t=\theta$ whenever observation is non-empty at the limit:
\begin{align*}
\lim_{t\rightarrow\infty}\mathcal{P}_{\sigma^*}(\hat{a}^t=\theta|B^t\neq\varnothing)=1.
\end{align*}
\end{defn}

Truth-telling observation is a weaker condition than asymptotic learning in two aspects. First, it requires only the state-matching action $\hat{a}^t$ to be perfectly correct \textit{conditional on non-empty observation} as $t\rightarrow\infty$, as opposed to the unconditional correct action in asymptotic learning. Second, even in an equilibrium with truth-telling observation, an agent's action conditional on non-empty observation may \textit{not} coincide with $\hat{a}^t$. This stems from coordination motives: when the community size is large, the agents may conform to an action that matches the state with a probability lower than $\frac{1}{2}$. In contrast, asymptotic learning requires each agent's equilibrium action to be always the same as $\hat{a}^t$ at the limit. Therefore, truth-telling observation should be regarded as a notion describing only the maximum informativeness of observation but not the correctness of equilibrium behavior, while asymptotic learning represents the highest level of both.

\section{Results for Exogenous Observation}

In this section, I present the main results for exogenous observation. A well-established theoretical prediction in much of the literature, which typically assumes that only one agent moves in each period, is that herding occurs when private beliefs are bounded: with a positive probability, all agents ultimately choose the wrong action after a particular time threshold. This occurs because learning cannot be improved indefinitely: at a certain point, some agent's observation becomes so informative that she abandons her private signal altogether and herds with her predecessors, and by this time, social learning essentially ceases. In this section, I will demonstrate how coordination motives can prevent herding, can incentivize agents to use their private information, and can lead to a learning level that can be arbitrarily close to asymptotic learning.

\subsection{Two Conditions on Observation Structure}

When observation is exogenous, the observation structure--indicating which predecessors each agent observes--plays an important role in determining whether asymptotic learning is possible. This structure is sometimes referred to as a \textit{network} in the literature to highlight the connection between theory and application. To illustrate how the observation structure influences learning, I provide several typical observation structures below.

\begin{itemize}
\item{1.} $B^t=N^1$ for all $t$: a ``star network'' in which each agent observes only the actions in the first community in the action sequence.
\item{2.} $B^t=N^{t-1}$: a ``line network'' in which each agent observes only the closest community.
\item{3.} $B^t=\cup_{i=1}^{t-1}N^i$: a ``complete network'' in which each agent observes every predecessor. This is the upper bound of observational information that can be obtained.
\end{itemize}

In the first structure, asymptotic learning is never possible regardless of private beliefs and community size because learning cannot be improved beyond the second period: all agents after period $1$ are essentially identical in terms of information acquired because they have identical observation. In the second and third observation structures, if each community is a singleton, then asymptotic learning occurs when private beliefs are unbounded but never occurs when private beliefs are bounded. Herding occurs in the latter case with a probability bounded away from zero.

I now introduce two conditions on the observation structure that lead to approximate asymptotic learning in the presence of coordination motives, \textit{regardless of whether private beliefs are unbounded.} The first condition stands in contrast to the ``star network'' above.

\begin{defn}
An observation structure has \textbf{expanding observations} if, for every $K\in\mathbb{N}^+$, $\lim_{t\rightarrow\infty}O^t(\max\{\tau:Q^{\tau}\subset B^t\}<K)=0$.
\end{defn}

The concept of expanding observations for singleton communities was first introduced by Acemoglu et al.\cite{ADLO}, and in my paper, I generalize this concept to the case with non-singleton communities. This approach implies that an agent always observes a predecessor who is not too far away. The ``star network'' clearly does not have expanding observations, but the ``line network'' and ``complete network'' do. This property on observation structure ensures that if any improvement on learning is ever possible, it is transmitted to \textit{every} agent over time via observation in the sense that no agent will be blocked from any recent development in learning by only observing some distant predecessors.

However, expanding observations alone is not sufficient for asymptotic learning when private beliefs are bounded. For instance, the ``line network'' has expanding observations, but with bounded private beliefs, herding occurs with positive probability regardless of how large the community becomes. Therefore, I introduce the second condition.

\begin{defn}
An observation structure has \textbf{infinite complete observations} if there exists a infinite subset of time periods $\{t_1,t_2,\cdots\}$ such that (1) for every $K\in\mathbb{N}^+$, $\lim_{n\rightarrow\infty}O^{t_n}(|B^{t_n}|<K)=0$, and (2) $\lim_{n\rightarrow\infty}O^{t_n}(\bar{B}^{\tau}\subset B^{t_{n}}\text{ }\forall N^{\tau}\subset B^{t_n})=1$.
\end{defn}

The meaning of infinite observations is straightforward. Complete observations indicate the existence of a subset of agents such that an agent in this subset who observes a predecessor also observes all actions that can possibly be observed by the predecessor. At the limit, an agent's observed neighborhood contains infinitely many actions. To understand why this condition is needed for approximate asymptotic learning with bounded private beliefs, first consider an observation structure that has only finite observations, such as a ``line network''. On the one hand, bounded private beliefs prevent private information within any finite neighborhood from being arbitrarily informative about the true state. On the other hand, finite observations and bounded private beliefs together imply that once any observation by any predecessor in this neighborhood becomes sufficiently informative, it cannot be improved upon by substituting the predecessor's action. To approach asymptotic learning, the only possibility is to have infinitely many actions in an observed neighborhood.

Next, consider the case in which an agent has an ``incomplete'' observation of a predecessor. Hence, there are some actions that cannot be observed by the agent but that may have been observed by the predecessor. Now the predecessor's action has an ambiguous effect on the agent's posterior belief because the agent needs to consider those unobserved actions and make a corresponding inference based on her observation. As a result, the direction of her posterior belief--whether it favors state $0$ or $1$ after observing the predecessor's action--cannot be determined solely by the predecessor's action but depends on her observation. In other words, observing action $1$ by the predecessor may cause the agent to favor either state in different situations, which makes her updated belief intractable. In contrast, complete observations determine the direction of the agent's posterior belief without ambiguity. I provide further detail below when presenting the formal result.

A simple example of an observation structure with both expanding observations and infinite complete observations is the ``complete network'' in which every agent observes the entire action history. In general, the two conditions are satisfied by a wide class of observation structures.

\subsection{Main Result}

I now present the main theoretical result of this section.

\begin{thm}
Assume that the observation structure has expanding observations and infinite complete observations. There exists $\hat{Q}$ such that if $G(Q\geq \hat{Q})>0$, for every $\epsilon>0$, there exists an equilibrium $\phi^*$ such that (1) truth-telling observation occurs and (2) $\lim_{t\rightarrow\infty}\mathcal{P}_{\phi^*}(a_n^t=\theta)>1-\epsilon$.
\end{thm}

This result shows that coordination motives can serve as an economic force that counters the herding incentive in a way that hurts individual agents \textit{ceteris paribus} but benefits social learning. When the community size is large, the signal can be regarded as a correlating device to coordinate agents in the same community to conform to an action based on the signal value alone. This action may sometimes differ from the more ``informed'' action based on both signal and observation, but it does constitute mutual best responses and makes the actions of this community informative for successors. This is the key difference between a model with coordination motives and a model without them: in the latter, because every agent always seeks to maximize her probability of matching the state, herding can never be prevented when private beliefs are bounded.

Following the rough intuition above, I present a heuristic proof of Theorem 1 (the complete proof with technical details can be found in the Appendix). First, properties of Bayes' update determine that regardless of which is the true state, an agent's posterior probability on the \textit{wrong} state can never become arbitrarily close to $1$ over time because otherwise the same set of observation inducing this posterior probability would need to occur with $>1$ probability when the true state is altered, which is a contradiction.

Next, I construct an equilibrium in which each observed action is informative. Consider an action profile that follows observation--that is, choose the action matching the state with higher probability given observation only--when the signal is weak and that follows the signal when the signal is strong. This situation constitutes mutual best responses when $Q^t$ is large because the incentive to conform becomes stronger than the incentive to match the state. In this equilibrium, strong private signals are never abandoned. As a result, for an agent who has complete observation of another agent following such an action profile, Bayes' update from observing this additional action will induce a posterior belief in favor of the corresponding state, in contrast to the belief that occurs without adding this observation. This claim implies a more important property of equilibrium behavior: following any belief about the state, \textit{additional} observation of sufficiently many actions of the same value can induce a new belief that entails a higher ($>\frac{1}{2}$) probability of the corresponding state.

Now we can demonstrate the truth-telling nature of observation. Consider a subset of agents with infinite complete observations, and note that the hypothetical action $\hat{a}^t$ can be regarded as the \textit{optimal} action for some outside singleton agent who observes $B^t$ and attempts to maximize her probability of matching the state. Suppose that truth-telling observation does not occur, which implies that her highest learning probability is equal to some $\rho<1$. Fix a sufficiently large $t'$ such that observing $B^{t'}$ gives her a $\approx\rho$ probability of matching the state, and consider another sufficiently large number $\Delta$ and the following \textit{sub-optimal} strategy: given the action sequence in $B^{t'}$, she will change her action if and only if she observes $\Delta$ consecutive \textit{additional} actions that are the same value, which is opposite of the action that she would have taken by observing only $B^{t'}$. It can be shown that this sub-optimal strategy already improves her learning probability by a significant amount, which makes the total probability exceed $\rho$--thus revealing a contradiction. Notably, the result is \textit{not} obtained by the law of large numbers, as observed actions are not mutually independent: later actions are affected by earlier actions via observational learning. Instead, this strict improvement stems from calculating the difference between the probabilities of the $\Delta$ actions being ``helpful'' (in the sense that they correct a wrong belief) and ``harmful'' (in the sense that they mislead from a correct belief); details are provided in the Appendix.

Finally, I identify a direct inverse relation between the limit learning probability and the probability of agents acting according to signal only. Truth-telling observation implies that at the limit, the probability of taking the correct action conditional on non-empty observation is equal to $1$; hence, the total learning probability at the limit is the sum of the probability that agents consider their observation and the probability that a strong signal occurs favoring the true state. The cutoff for a strong signal is arbitrary--as long as each agent uses her signal for a fixed positive probability, truth-telling observation occurs. Hence, the higher this cutoff is, the more likely an agent chooses her action according to observation, and thus, the higher the learning probability is. In this way, any learning probability that is less than $1$ can be obtained in equilibrium.

Note that the condition of infinite complete observations is not required for asymptotic learning when private beliefs are unbounded. This has been proved for singleton communities in the literature (see, e.g., Smith and Sorensen\cite{SS}) and is extended to this model with coordination motives. I note this result below and use it in subsequent proofs.

\begin{prop}
Assume that the observation structure has expanding observations and that private beliefs are unbounded. There always exists an equilibrium with asymptotic learning.
\end{prop}

\section{Results on Endogenous Observation}

In this section, I analyze the model under endogenous observation by discussing the cases of unbounded and bounded private beliefs separately. Note that costly and strategic observation creates an independent economic force by itself: it discourages an agent from observation when her signal is informative because the additional benefit from observation becomes small or even negligible. With this added strategic component, the effect of coordination motives becomes more subtle, but in general, a similar implication can be derived: with sufficiently strong coordination motives, the level of social learning can be improved.

\subsection{Unbounded Private Beliefs}

\subsubsection{Singleton Communities}

To fully understand how coordination motives change the pattern of learning, it is important to first understand how singleton agents--that is, $G(1)=1$--behave when observation is endogenous; this behavior has scarcely been explored in the previous literature. First, I show that asymptotic learning never occurs in this environment.

\begin{subtheorem}{thm}
\begin{thm}[Song (2015)]
Asymptotic learning does not occur in any equilibrium.
\end{thm}

Although the result stands in stark contrast to the well-known result in the literature with exogenous observation, the underlying argument for this result is rather straightforward. By assumption, no private signal perfectly reveals the true state; for asymptotic learning to occur in any equilibrium, a necessary condition is that an agent chooses to observe for almost every private signal. If an agent chooses to observe, her payoff is upper bounded by $u(1,1,1)-c$ because the best possible observation is one that fully reveals the true state and guarantees her a benefit of $u(1,1,1)$. If she chooses not to observe and simply follows her private signal $s$, her expected payoff is equal to $\frac{\max\{f^1_1(s),f^1_0(s)\}}{f^1_1(s)+f^1_0(s)}u(1,1,1)+(1-\frac{\max\{f^1_1(s),f^1_0(s)\}}{f^1_1(s)+f^1_0(s)})u(1,0,1)$. Because private beliefs are unbounded, for any positive $c$, there is always a positive measure of signals such that the payoff from no observation is higher, which implies that asymptotic learning never occurs.

Although asymptotic learning is impossible, efficient information aggregation can still be achieved in the form of truth-telling observation. Assuming a symmetric signal structure, the following result provides a necessary and sufficient condition for truth-telling observation in this environment as well as a full characterization of the limit learning probability.

\begin{thm}[Song (2015)]
Assume that the signal structure is symmetric: $f^1_0(s)=f^1_1(-s)$ for every $s\in S$. Truth-telling observation occurs in $\sigma^*$ if and only if $\lim_{t\rightarrow\infty} K(t)=\infty$. $\lim_{t\rightarrow\infty}\mathcal{P}_{\sigma^*}(a^t=\theta)=F^1_0(s^*)$ where $s^*$ is characterized by
\begin{align*}
\frac{f^1_1(s^*)}{f^1_0(s^*)+f^1_1(s^*)}u(1,1,1)+\frac{f^1_0(s^*)}{f^1_0(s^*)+f^1_1(s^*)}u(1,0,1)=u(1,1,1)-c.
\end{align*}
\end{thm}
\end{subtheorem}

%The generic uniqueness of $\sigma^*(1)$ is obtained by an inductive argument: starting from period $1$, each agent faces a discrete choice in observation as well as in action. Since the agent's objective is to maximize her probability of taking the correct action, in general there is a unique solution to the optimal decision in both. Proceeding inductively, the unique equilibrium can be determined.

The property of truth-telling observation also holds in the case of exogenous observation with unbounded private beliefs and expanding observations, but the underlying mechanism here is much different. Under exogenous observation, an agent always uses her private information with a positive probability (which converges to 0 over time) because her signal can be strong enough to overwhelm the \textit{realized} observation. Under endogenous observation, an agent may choose to use her private information and \textit{not observe at all} because although observation can still be beneficial, its marginal benefit in information does not cover the cost. This probability of no observation does \textit{not} converge to 0 over time. As a result, an agent's individual action is always erroneous with a probability bounded away from 0, which then implies that observing any finite sequence of actions does not reveal the true state regardless of when the actions occurred. In other words, truth-telling observation never occurs when $\lim_{t\rightarrow\infty} K(t)\neq\infty$. However, this individual error is precisely the source of informativeness: because an agent sometimes chooses to forgo the (potentially more informative) observation, her action is indicative of the range of signals that she receives. Therefore, once an agent observes an arbitrarily large neighborhood, information can be aggregated efficiently to reveal the true state. Once again, this follows \textit{not} from the law of large numbers but from an argument of continuing strict improvement similar to that in Theorem 1.

In terms of the limit learning probability, it is straightforward that $F^1_0(s^*)$ is the largest possible learning probability in equilibrium, and it is achievable only when truth-telling observation occurs. After all, it is impossible in any equilibrium for any agent to choose to observe when her signal is not in $[-s^*,s^*]$. Hence, we can conclude that with unbounded private beliefs, endogenous observation lowers the limit learning probability compared with the most informative scenario under exogenous observation. However, endogenous observation may lead to a higher limit learning probability than exogenous observation without expanding observations because although agents will not observe the given extreme signals, they make more informed choices when they do observe. For instance, consider the ``star network'' in the previous example of observation structures. It can be shown that if observation is endogenous and $K(t)=1$, each agent will observe her immediate predecessor whenever she chooses to observe, and the limit learning probability is higher than that in the ``star network'' when $c$ is low.

\subsubsection{Non-Singleton Communities}

In this section, I present the main result for non-singleton communities and compare it with the result for singleton communities above.

\begin{thm}
There is a cutoff $\bar{c}>0$ such that for every $c\in(0,\bar{c})$, there exists $\hat{Q}(c)$ such that if $G(Q\geq \hat{Q}(c))=1$, there exists an equilibrium $\sigma^*$ with asymptotic learning.
\end{thm}

Before elaborating on this result, I must first describe such an equilibrium that leads to asymptotic learning. For agent in the same community $N^t$, consider two action profiles: a ``truth-seeking'' profile in which agents conform to the action that matches the state with higher probability according to all available information and a sub-optimal profile in which they act otherwise--for example, they conform to the action that is the worst match for the state. The first profile clearly yields a higher payoff for every agent in expectation. Now consider the following strategy profile for observation and action: agent 1 observes a prescribed neighborhood given $s^t(Q^t)$, and no other agent observes. If the realized observed neighborhood is revealed to be the same as the prescribed one, then the agents follow the ``truth-seeking'' action profile; otherwise, they follow the sub-optimal profile.

When the community size is large, both action profiles constitute best responses, which then implies (by backward induction) that making the prescribed observation is indeed optimal for agent $1$. Hence, we have an equilibrium in which the observation of an arbitrary non-empty neighborhood occurs regardless of signal value. By imposing the property of expanding observations on this sequence of the observed neighborhood (for example, agent 1 in each period observes agent 1 in the previous period), we can apply Proposition 2 to obtain asymptotic learning.

This result identifies an effect on strategic observation that is imposed by coordination motives: more observation can be encouraged as the community size grows. In the equilibrium described above, by conforming to different actions according to the observed neighborhood, the agents essentially make it more costly for agent 1 not to observe, and thus, the range of signals for which agent 1 will observe the prescribed neighborhood expands. When the community size becomes sufficiently large, this signal range becomes the whole support $S$, and hence, an unbroken chain of observation is established even when observation is costly. As a result, the efficient aggregation of information is restored.

From the construction of equilibrium, we can also see that the result is robust to the specific cost structure of observation. In a more general model, let $c^t(k)$ denote a cost function for observing $k$ predecessors in period $t$. As long as $c^t(1)$ has a constant upper bound, Theorem 3 can be applied to show that asymptotic learning can occur in equilibrium.

\subsection{Bounded Private Beliefs}

When private beliefs are bounded and only a finite-size neighborhood can be observed at the limit, the level of social learning is always bounded away from 1 because of either herding or a persistent probability of error\footnote{This claim is valid for both exogenous and endogenous observation. Formal results can be found in Song\cite{Song}.}. Therefore, in this section, I assume that $\lim_{t\rightarrow\infty}K(t)=\infty$ to show a sharp contrast between learning with and without coordination motives.

As in the previous section, I first discuss the effect of endogenous observation on learning in an environment in which agents are singletons. The limit learning probability can be affected in either direction: whether it rises or falls compared with exogenous observation depends greatly on the value of $c$, the cost of observation. The following example illustrates this result and its underlying mechanism without loss of generality.

Assume that only one agent moves in each period. Consider the following two cases: exogenous observation where $B^t=\{1,\cdots,t-1\}$ and endogenous observation where $K(t)=t-1$. It is established in the literature (see, e.g., Smith and Sorensen\cite{SS}) that when observation is exogenous, the limit learning probability has an upper bound $\bar{P}<1$. In other words, at the limit, an agent behaves better than merely following her own signal, but she cannot learn the true state perfectly.

Under endogenous observation, Theorem 2 can be extended here to characterize the limit learning probability for a range of the cost $c$. Consider a symmetric signal structure. Note that unbounded private beliefs constitute a sufficient but not necessary condition for the proof of Theorem 1. In fact, truth-telling observation requires only that beliefs be ``strong'' relative to cost, i.e., $\lim_{s\rightarrow 1}\frac{f^1_1(s)}{f^1_1(s)+f^1_0(s)}\times u(1,1,1)+\lim_{s\rightarrow 1}\frac{f^1_0(s)}{f^1_1(s)+f^1_0(s)}\times u(1,0,1)>u(1,1,1)-c$. In other words, as long as an agent prefers not to observe--even if observation reveals the truth--when her signal takes the most extreme value, the necessary and sufficient relation between truth-telling observation and infinite observation at the limit can be derived following the same argument as used previously. Hence, when $c>\lim_{s\rightarrow 1}\frac{f^1_0(s)}{f^1_1(s)+f^1_0(s)}\times(u(1,1,1)-u(1,0,1))$, letting $s(c)$ be characterized by $\frac{f^1_1(s(c))}{f^1_1(s(c))+f^1_0(s(c))} u(1,1,1)+\frac{f^1_0(s(c))}{f^1_1(s(c))+f^1_0(s(c))} u(1,0,1)=u(1,1,1)-c$, we have an expression for the limit learning probability that is denoted $P(c)$:
\begin{align*}
P(c)=F^1_0(s(c)).
\end{align*}

Depending on $c$, the value of $s(c)$ ranges from $0$ to arbitrarily close to $1$. As a result, the value of $P(c)$ ranges from $F^1_0(0)$ to arbitrarily close to $1$. We see here that endogenous observation affects social learning in a way that is monotonic in $c$: compared with exogenous observation, endogenous observation is better for social learning when $c$ is relatively large and is worse for social learning when $c$ is relatively small.

I now present the main result for coordination motives. Regardless of the value of $c$, coordination motives facilitate learning in the sense that it increases the highest possible equilibrium learning probability.

\begin{thm}
There is a cutoff $\bar{c}>0$ such that for every $c\in(0,\bar{c})$, there exists $\hat{Q}(c)$ such that if $G(Q\geq \hat{Q}(c))=1$, for every $\epsilon>0$, there exists an equilibrium $\sigma^*$ where (1) truth-telling observation occurs and (2) $\lim_{t\rightarrow\infty}\mathcal{P}_{\sigma^*}(a_n^t=\theta)>1-\epsilon$.
\end{thm}

This result can be derived from a combination of Theorems 1 and 3. First, based on Theorem 3, agents can be incentivized to observe a prescribed neighborhood given any signal; then, according to Theorem 1, when the prescribed neighborhood is observed, the signal serves as a correlation device for the agents to coordinate on an action profile that accounts for all available information with a certain probability. This probability can be arbitrarily close to $1$. Consequently, for any fixed observation cost $c$, when the community size is large, there is always an equilibrium with a higher learning probability than is available for singleton communities.

\subsection{Summary}

Before discussing some extensions of the model, I briefly summarize the comparison across observation structures and community sizes in this section. To introduce a different and useful perspective for examining the impact of various factors on social learning, I categorize the main results here by signal structure and regard the case with exogenous observation and singleton communities as a benchmark.

When private beliefs are unbounded, in the benchmark case, the level of social learning depends entirely on the pattern of observation. Asymptotic learning occurs if and only if at the limit an agent almost surely observes a close predecessor (e.g., the ``complete'' network). The presence of coordination motives does \textit{not} change this property of learning. When observation becomes endogenous, asymptotic learning cannot be achieved because the positive observation cost prevents an agent from observing when her signal is strong. Imposing coordination motives now makes a difference in the sense that it encourages observation and thus restores asymptotic learning when the community size is sufficiently large. Figure 1 uses some representative observation structures to illustrate the learning pattern over time in different environments.

\begin{figure}[h]
\centering
\includegraphics[width=5in]{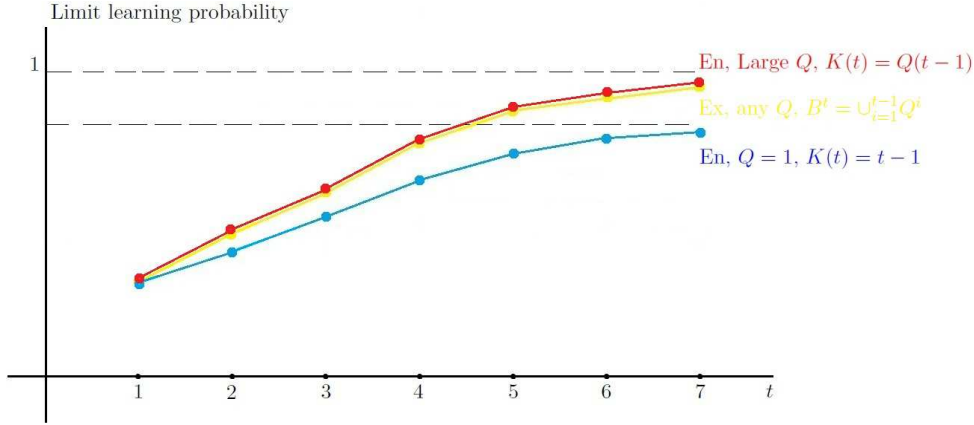}
\caption{Learning Patterns with Unbounded Private Beliefs}
\setlength{\abovecaptionskip}{0pt}
\caption*{(En = endogenous observation; Ex = exogenous observation)}
\end{figure}

When private beliefs are bounded, the benchmark case typically produces a learning probability bounded away from $1$, regardless of whether agents observe close or distant predecessors. Making observation endogenous can cause this probability to be either higher or lower depending on the observation cost $c$. With coordination motives, the highest possible learning probability increases for any value of $c$ when the community size is sufficiently large; in particular, it can be arbitrarily close to $1$ in equilibrium. Figure 2 illustrates these scenarios.

\begin{figure}[h]
\centering
\includegraphics[width=5in]{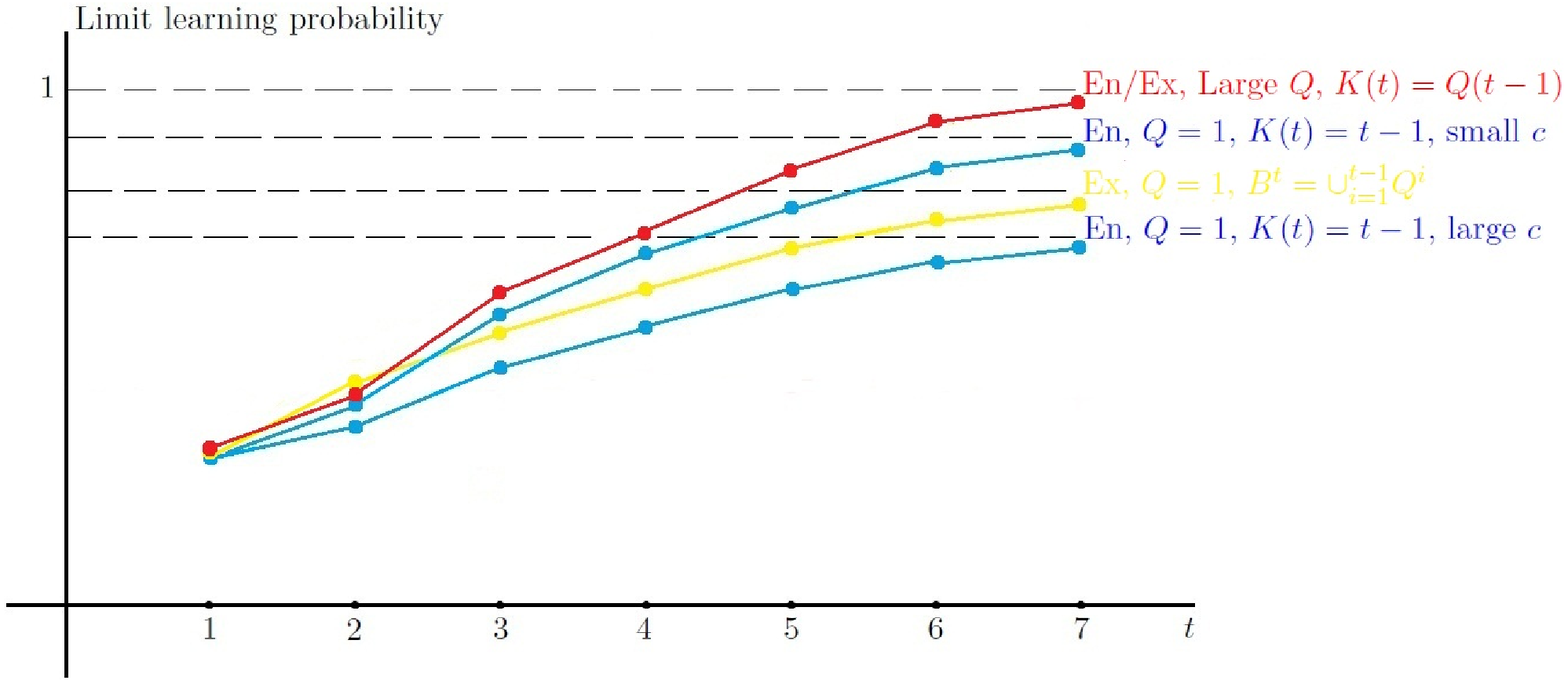}
\caption{Learning Patterns with Bounded Private Beliefs}
\setlength{\abovecaptionskip}{0pt}
\caption*{(En = endogenous observation; Ex = exogenous observation)}
\end{figure}

\section{Discussion}

%\subsection{Random Community Size}

%In many applications, the community size $Q$ is not constant over time. In this section, I demonstrate how the model can be generalized to account for a more variable environment with random community size.

%Instead of a fixed $Q$, assume that at the beginning of each period, the community size $Q(t)$ is randomly selected from a commonly known probability distribution $H$ on $\mathbb{N}^+$, with $\mathbb{E}[Q(t)]<\infty$. The $Q(t)$'s are independent and identically distributed over time. $Q(t)$ is common knowledge for agents in $Q^t$ before they receive $s^t$. I refer to this environment as the \textit{random community size model}. Whether $Q(t)$ can be observed by agents after period $t$ or not, the main results derived before generalize to this model.

%\begin{prop}
%There exists $\hat{Q}$ such that if $H(Q(t)\geq \hat{Q})>0$, Theorems 1, 3 and 4 hold in the random community size model.
%\end{prop}

%This result implies that, for coordination motives to have all the previously described effects on social learning, all the communities do not have to be large. Instead, it suffices to have infinitely many communities of large size over time. The intuition for this generalization is that as long as agents use some of their private information conditional on a large community, their action still conveys some valuable information to successors, whether or not their realized community size can be observed.

\subsection{Equilibrium Selection and Risk Dominance}

The conforming incentive generated by coordination motives results in multiple equilibria in an environment with large communities. Many of my previous results are built on the fact that conforming based on the most informed action and on a less informed action are both optimal responses for agents in the same community, which does not occur when agents are singletons because one's unique best response would then be to use all available information. A natural question is then whether different equilibria can be compared in any way and, if so, whether a selected equilibrium by any criterion changes the implication of having coordination motives in the model. In this section, I propose \textit{risk dominance} as an equilibrium selection method and discuss its properties and impact. In particular, this criterion is imposed on the interim stage in which signal and observation have been realized: it essentially enables comparison between two action profiles and selects a unique equilibrium action for each information set.

Consider any $N^t$ with any community size $Q^t$ and any information set $I^t$. Let $a^t(I^t)=\{a^t_{n}(I^t)\}_{n=1}^{Q^t}$ and $a^{'t}(I^t)=\{a^{'t}_{n}(I^t)\}_{n=1}^{Q^t}$ denote two arbitrary action profiles with unanimous action. Let $v_n^t(a^t(I^t),I^t)$ denote agent $n$'s expected payoff given $a^t(I^t)$ and $I^t$.

\begin{defn}
I say that $a^t(I^t)$ \textbf{risk-dominates} $a^{'t}(I^t)$ if for every $N^{'t}\subset N^t$ and every $n\in N^{'t}$, we have
\begin{align*}
&v_n^t(a^t(I^t),I^t)-v_n^t((\{a_i^{'t}(I^t)\}_{i\in Q^{'t}},\{a_j^{t}(I^t)\}_{j\notin N^{'t}}),I^t)\\
\geq &v_n^t(a^{'t}(I^t),I^t)-v_n^t((\{a_i^{t}(I^t)\}_{i\in Q^{'t}},\{a_j^{'t}(I^t)\}_{j\notin N^{'t}}),I^t).
\end{align*}
If $a^t(I^t)$ risk-dominates every other action profile for every $I^t$, we say that $a^t(I^t)$ is \textbf{risk dominant}.
\end{defn}

The idea behind risk dominance is the following: suppose that a subset of agents $N^{'t}\subset N^t$ switches their action from a given profile to an alternative profile. If action profile 1 risk-dominates action profile 2, then the \textit{expected loss} for every agent in $N^{'t}$ in switching from action profile 1 to 2 is always larger than the loss involved in switching from 2 to 1, for every possible $N^{'t}$ and information set $I^t$. One interpretation of risk dominance is that it indicates an agent's preference for one action profile over the other when she is uncertain about which will be used by others in her community.

An intuitive candidate for a risk-dominant action profile is one in which each agent makes the best use of $I^t$, for which I provide a formal definition below. This is actually the generically unique risk-dominant action profile.
\begin{defn}
An action profile is \textbf{truth-seeking} if for every $t,n$ and $I^t$, $n$ chooses the action $a_n^t$ that maximizes the probability of $a_n^t=\theta$ given $I^t$.
\end{defn}

\begin{prop}
The truth-seeking action profile is risk dominant.
\end{prop}

It is easy to see that the truth-seeking action profile yields the highest possible expected payoff for every agent in $N^t$ given $I^t$. In fact, if we fix the number of agents that take a given action $a$, each agent's payoff is the highest when $a$ is truth-seeking. Hence, given a subset $N^{'t}$ of action-switching agents, their loss is always less when switching from some other action profile to the truth-seeking profile rather than the opposite. I name the equilibrium with the truth-seeking profile as a \textit{truth-seeking equilibrium} and denote it as $\phi^*_T$ and $\sigma^*_T$ under exogenous and endogenous observation, respectively. I then inspect how learning in this particular equilibrium changes according to the community size.

When observation is exogenous, coordination motives have \textit{no effect} on the general learning pattern in a truth-seeking equilibrium: regardless of $G$, asymptotic learning occurs if private beliefs are unbounded and if the observation structure has expanding observations and does not occur otherwise\footnote{To be precise, it has been proven that asymptotic learning does not occur when the observation structure does not have expanding observations or when private beliefs are bounded and the observation structure takes several typical forms. For a more specific account, see, e.g., Acemoglu et al.\cite{ADLO} and Song\cite{Song}.}. The truth-seeking action profile prevents agents from conforming to a less informed action that uses more of their private information, and hence, the whole community acts as a single agent that makes her best effort to match her action with the true state. The conforming incentive does not alter anything in the agents' behavior regardless of how large a community becomes.

When observation is endogenous, however, coordination motives still play an important role on learning in a truth-seeking equilibrium. For example, consider the following scenario. Suppose that the signal distribution is the same for every $Q$, i.e., $f_a^Q=f_a$, $a=0,1$. $\{f_0,f_1\}$ is symmetric with unbounded private beliefs. In addition, the agents have infinite observations ($\lim_{t\rightarrow\infty}K(t)=\infty$).

\begin{prop}
For every $\sigma^*_T$, we have $\lim_{t\rightarrow\infty}\mathcal{P}_{\sigma^*_T}(a_n^t=\theta)=\mathbb{E}_G[F_0(s^*(Q))]$, where $s^*(Q)$ is characterized by the following equation:
\begin{align*}
\frac{f_1(s^*(Q))}{f_0(s^*(Q))+f_1(s^*(Q))}u(1,1,Q)+\frac{f_0(s^*(Q))}{f_0(s^*(Q))+f_1(s^*(Q))}u(1,0,Q)=u(1,1,Q)-c.
\end{align*}
\end{prop}

Given the probability distribution of community size $G$, let $P_T(G)$ denote the probability of the correct action occurring at the limit in any truth-seeking equilibrium. We have the following corollary:

\begin{cor}
For $G',G$ such that $G'$ first-order stochastically dominates $G$, $P_T(G')>P_T(G)$.
\end{cor}

In a truth-seeking equilibrium, coordination motives still encourage agents to observe--not because no observation or a ``wrong'' observation entails that they conform to a sub-optimal action as in the previous results but because observation generates greater expected benefit in a larger community. As a result, the range of signals leading to non-empty observation increases while the truth-telling property of observation is preserved. Therefore, a larger community size increases the limit learning probability, while the incremental improvement decreases in $\sigma^*_T$ compared with the constructed equilibrium in Section 5 because asymptotic learning does not occur.

When private beliefs are bounded, coordination motives can work in opposite directions. As argued in Section 5, truth-telling observation occurs in $\sigma^*_T$ whenever private beliefs are ``strong'' relative to cost, i.e., when the payoff of simply following an extreme signal exceeds that of knowing the true state by costly observation. Similar to the above proposition, the first payoff can be written as $\frac{f_1(s^*)}{f_0(s^*)+f_1(s^*)}u(1,1,Q)+\frac{f_0(s^*)}{f_0(s^*)+f_1(s^*)}u(1,0,Q)$, while the second payoff is $u(1,1,Q)-c$, which implies that the marginal effect of increasing $Q$ is higher in the latter. We can then conclude that increasing the likelihood of a larger community size is better for social learning when private beliefs remain ``strong'' because, once again, it encourages observation that is still truth-telling. However, social learning is harmed when private beliefs become weak because the informativeness of observation may overwhelm that of any private signal and induce herding.

\subsection{Herding with Imperfect Knowledge of Signals}

The previous analysis has shown that coordination motives can reduce herding and facilitate efficient aggregation of private information. An important assumption leading to this result is that all information -- both private signals and observations -- is shared within community and can be used as a correlating device. In this section, I show that when private signals are ``private'' in the most strict sense and cannot be shared, coordination motives have exactly the opposite effect: herding occurs even when private beliefs are unbounded.

Assume that in every community $N^t$, each agent $n$ receives a private signal $s_n^t$ that cannot be observed by \textit{any} other agent. Without loss of generality, assume that $s_n^t$ follows independent and identical distribution $\{F_0,F_1\}$ conditional on the state, for every $t$ and $n$. Consider an exogenous ``line'' network, i.e. $B^t=N^{t-1}$. In every symmetric equilibrium, i.e. equilibrium where agents in the same community use identical strategies, I find the following property:

\begin{prop}
There exists $\hat{Q}$ such that if $G(Q\geq \hat{Q})=1$, asymptotic learning never occurs in any equilibrium.
\end{prop}

This result stands in stark contrast to the case of signal sharing, that there exists an equilibrium with asymptotic learning whenever private beliefs are unbounded and some close predecessors are observed. The reason behind is that keeping signals private brings potential uncertainty about other agents' actions. To achieve asymptotic learning, an agent's action must at least \textit{sometimes} contradict her observation (presumably when her signal is rather strong). However, when a highly informative observation realizes, since the observation is shared within the community and signals are independent, the agent can infer that no matter what the true state is, the other agents in her community will probably take the action that matches the observation. The existence of coordination motives then determines that her optimal action is to follow the observation as well. As a result, herding occurs with positive probability in every equilibrium even when private beliefs are unbounded.

Herding in this environment highlights the importance of \textit{common beliefs} on asserting the impact of coordination motives. When private signals are shared, agents have identical beliefs which make it possible for them to coordinate on either action in equilibrium when their community is large. As a result, they may agree to rely on private information only even though their observation is rather informative already. However, once each private signal is observed by the corresponding agent only, beliefs become heterogeneous across agents. The impossibility to ascertain one another's private signal incentivizes every agent to choose the ``safer'' option of following the commonly known observation, which then leads to herding.

\subsection{Separation Motives}

Thus far, I have focused on scenarios in which agents are willing to conform to a unanimous action. However, in some cases, there may be a ``congestion effect'' on action, i.e., more agents choosing the same action results in a smaller payoff for each agent. For instance, too many customers squeezing into a restaurant will probably result in a negative dining experience in terms of waiting time and noise level even if the restaurant is superior to its competitors in food quality. Consequently, a customer may actually prefer another restaurant with ordinary food but smaller crowds.

In this section, I show how the model developed above can be used to analyze this opposite environment with separation motives and its impact on learning. Consider the payoff function $u(\theta,a_n^t,m)$, and replace the second assumption on $u$ with the assumption that $u$ is decreasing in $m$.

Assume that observation is exogenous. For community $N^t$, let $P$ denote an arbitrary posterior probability that the true state is $1$ given their signal and observation. The following result describes the pattern of equilibrium behavior:

\begin{prop}
If $P>\frac{1}{2}$, in every equilibrium $\phi^*$, at least half of the agents in $N^t$ choose action $1$.
\end{prop}

%In any equilibrium, the number of agents choosing action $1$, denoted $Q_1$, must satisfy
%\begin{align*}
%Pu(1,1,Q_1)+(1-P)u(1,0,Q_1)\geq&(1-P)u(1,1,Q-Q_1+1)+Pu(1,0,Q-Q_1+1)\\
%(1-P)u(1,1,Q-Q_1)+Pu(1,0,Q-Q_1)\geq&Pu(1,1,Q_1+1)+(1-P)u(1,0,Q_1+1).
%\end{align*}
%Combining these two inequalities, we have
%\begin{align*}
%\frac{Q_1}{Q-Q_1+1}\leq \frac{P\bar{u}+(1-P)\underline{u}}{P\underline{u}+(1-P)\bar{u}}\leq \frac{Q_1+1}{Q-Q_1}.
%\end{align*}

This result asserts that regardless of community size, at least half of the agents will always take the more informed action. Here, separation motives work exactly the opposite of coordination motives: instead of urging agents to conform to an action taken by the majority, separation motives divide agents into two groups, always with more agents in the group representing better information. This situation results from the trade-off between a less crowded group and a more informative action. Moreover, it can be shown that for some classes of utility function (e.g., $u$ is linear in $\frac{1}{m}$), as the community size rises, one can make increasingly precise inferences on the agents' posterior belief from observing all actions in the community. Then, if observation is exogenous and more or less ``complete'', i.e., at least at the limit, an agent observes almost the entire action history, the learning pattern is similar to that with singleton communities. Asymptotic learning occurs when private beliefs are unbounded but never occurs otherwise.

When observation is endogenous, a natural conjecture is that negative externalities discourage observation, and this possibility is confirmed by the model. As the community size increases, the marginal benefit from observation decreases because the equilibrium actions are always split in certain proportions between $0$ and $1$. Hence, although truth-telling observation still occurs if infinite observations can be made at the limit, the range of signals under which observation is non-empty is narrowed by separation motives. Moreover, a ``tragedy of commons'' argument implies that more precise knowledge about the true state may actually decrease the total payoff in a community, and hence, the issue of discrepancy between equilibrium and efficiency arises, but I will not discuss this issue further in this paper.

\section{Conclusion}

In this paper, I studied the problem of Bayesian learning with coordination motives in various signal and observation structures. A large and growing body of literature on social learning focuses on whether equilibria lead to efficient information aggregation, but most studies assume exogenous observation and no coordination motives. In many relevant situations, these two assumptions are overly simple. Individuals sometimes obtain their information not by some exogenous stochastic process but as a result of strategic choices. In addition, their payoffs may be directly affected by the actions of other individuals. This raises the questions of how different combinations of factors influence learning differently, under what circumstance asymptotic learning can be achieved, and how the results compare with benchmark cases studied in the literature.

To address these questions, I formulated a sequential-move learning model that incorporates all these elements. The basic decision sequence of the model follows the convention of Bikhchandani, Hirshleifer and Welch\cite{BHW}, Smith and Sorensen\cite{SS} and Acemoglu et al.\cite{ADLO}: on a discrete time line, a signal about the underlying binary state is realized at the beginning of every period and is observed by each agent in that period only. Each agent takes a binary action at the end of their period, and meanwhile, she can observe some of her predecessors' actions that are potentially informative. Nevertheless, my model differs from that used in most research in two fundamental aspects. First, in the literature, there is usually only one agent in each period, whereas in my model, there is a community consisting of multiple agents. Within a community, agents share their information (reflected by the signal) and observation and take their actions simultaneously. Also in contrast to the literature, in which each agent's sole objective is to match her action with the state, an agent's payoff from a given action is determined by both the state and the number of others in her community that take the same action. Second, observation is assumed to be exogenously given in much of the literature, whereas in this paper, I also analyze the case in which each agent can pay a cost to strategically choose a subset of her predecessors to observe.

I characterized pure-strategy (perfect Bayesian) equilibria for each observation structure (exogenous and endogenous) and characterized the conditions under which asymptotic learning can be obtained or approximated. When observation is exogenous, asymptotic learning occurs if private beliefs are unbounded and observations are ``expanding'', i.e., every observed neighborhood contains the action of some close predecessor over time. This result holds regardless of the community size. If private beliefs are bounded, for most common observation schemes, the probability of learning is bounded away from $1$ when the community size is small, but it can become arbitrarily close to asymptotic learning when the community size is larger than a certain threshold. Coordination motives reduce herding and improve social learning in this case.

When observation is endogenous, coordination motives also help to achieve better social learning but in a very different way. With a small community size, asymptotic learning never occurs because agents do not always observe: when the private signal is strong, it is not worthwhile to pay the observation cost for a small marginal expected benefit. However, when the community size becomes large, coordination motives encourage observation even when the private signal is strong because the marginal benefit from observation increases with the number of agents in a community. Therefore, asymptotic learning (or nearly asymptotic learning) occurs even when observation is costly.

I also discussed the issue of equilibrium selection and proposed risk dominance as a selection criterion for the action profile after both signal and observation are realized. In the selected equilibria, coordination motives do not affect learning at all when observation is exogenous, have a positive effect on learning when observation is endogenous and private beliefs are unbounded, and may either positively or negatively influence learning when observation is endogenous and private beliefs are bounded.

Beyond the specific results presented in this paper, I believe that the framework developed here can be applied to analyze learning dynamics in a more general and complex environment. The following questions are among those that can be studied in future work using this framework: (1) equilibrium learning when agents' preferences are heterogeneous, both over time and within a community; (2) the effect of coordination motives when agents in the same community make sequential decisions; and (3) equilibrium learning when the strength of coordination motives depends on the true state.

\newpage
\appendix
\section*{APPENDIX}

\begin{proof}[Proof of Proposition 1]
Suppose that there exist some $\sigma^*$ and $I^t$ such that in $N^t$, $Q'\in(1,Q^t)$ agents choose action $1$ and the others choose action $0$ in equilibrium. Let $P=\mathcal{P}_{\sigma}(\theta=1|I^t)$; for every agent that chooses action $1$, we have
\begin{align*}
Pu(1,1,Q')+(1-P)u(1,0,Q')\geq Pu(1,0,Q^t-Q'+1)+(1-P)u(1,1,Q^t-Q'+1).
\end{align*}
For every agent that chooses action $0$, we have
\begin{align*}
Pu(1,0,Q^t-Q')+(1-P)u(1,1,Q^t-Q')\geq Pu(1,1,Q'+1)+(1-P)u(1,0,Q'+1).
\end{align*}
Combining the inequalities yields
\begin{align*}
Pu(1,1,Q')+(1-P)u(1,0,Q')\geq Pu(1,1,Q'+1)+(1-P)u(1,0,Q'+1),
\end{align*}
which is a contradiction.
\end{proof}

\begin{proof}[Proof of Theorem 1]

I prove the result by establishing several lemmas.

\begin{lem}
There exists $\hat{Q}$ such that for any $Q^t\geq \hat{Q}$ and for any $I^t$, every action profile with unanimous action constitutes mutual best responses in $N^t$.
\end{lem}
\begin{proof}
Without loss of generality, assume that every agent in $N^t$ chooses action 1 given $I^t$. Let $P$ denote the probability that $\theta=1$ given $I^t$. For each agent, her expected payoff from action $1$ is $Pu(1,1,Q^t)+(1-P)u(1,0,Q^t)$, while her payoff from action $0$ is $Pu(1,0,1)+(1-P)u(1,1,1)$. For any $P\in(0,1)$, as long as $Q^t$ is such that $u(1,0,Q^t)\geq u(1,1,1)$, the agent's expected payoff from action 1 is higher. Hence, the action profile with unanimous action constitutes mutual best responses.
\end{proof}

Given an equilibrium $\phi^*$, consider an arbitrary subset of the set of time periods with infinite complete observations, which consists of sufficiently many consecutive communities of at least size $\hat{Q}$ starting from some time period $k$. From the assumption that $G(Q\geq \hat{Q})>0$, such a subset exists with probability 1. Let $B_k$ be the neighborhood consisting of the first $k$ communities, and consider any agent in a community of size $Q$ who has a private signal $s$ and observes $B_k$. Let $R_{\phi^*}^{B_k}$ be the random variable of the posterior belief about the true state being $1$ given each decision in $B_k$. For each realized belief $R_{\phi^*}^{B_k}=r$, we say that a realized private signal $s$ and decision sequence $h$ in $B_k$ \textit{induce} $r$ if $\mathcal{P}_{\phi^*}(\theta=1|h,s)=r$.

\begin{lem}
In any equilibrium $\phi^*$, for either state $\theta=0,1$ and for any $s\in S$, we have
\begin{align*}
&\lim_{\epsilon\rightarrow 0^+}(\lim\sup_{k\rightarrow\infty}\mathcal{P}_{\phi^*}(R_{\phi^*}^{B_k}>1-\epsilon|0,s))\\
=&\lim_{\epsilon\rightarrow 0^+}(\lim\sup_{k\rightarrow\infty}\mathcal{P}_{\phi^*}(R_{\phi^*}^{B_k}<\epsilon|1,s))=0.
\end{align*}
\end{lem}

\begin{proof}
I prove here that $\lim_{\epsilon\rightarrow 0^+}(\lim\sup_{k\rightarrow\infty}\mathcal{P}_{\phi^*}(R_{\phi^*}^{B_k}>1-\epsilon|0,s))=0$, and the second equality would follow from an analogous argument. Suppose that the equality does not hold; then $s\in S$ and $\rho>0$ exist such that for any $\epsilon>0$ and any $N\in\mathbb{N}$, $k>N$ exists such that $\mathcal{P}_{\phi^*}(R_{\phi^*}^{B_k}>1-\epsilon|0,s)>\rho$. Consider any realized action sequence $h_{\epsilon}$ from $B_k$ that, together with $s$, induces some $r>1-\epsilon$, and let $H_{\epsilon}$ denote the set of all such action sequences; thus, we know that
\begin{align*}
\frac{\mathcal{P}_{\phi^*}(h_{\epsilon}|\theta')f^Q_{\theta'}(s)}{\mathcal{P}_{\phi^*}(h_{\epsilon}|\theta)f^Q_{\theta}(s)+\mathcal{P}_{\phi^*}(h_{\epsilon}|\theta')f^Q_{\theta'}(s)}&=r\\
\sum_{h_{\epsilon}\in H_{\epsilon}}\mathcal{P}_{\phi^*}(h_{\epsilon}|\theta)&>\rho.
\end{align*}
The above two conditions imply that
\begin{align*}
1\geq \sum_{h_{\epsilon}\in H_{\epsilon}}\mathcal{P}_{\phi^*}(h_{\epsilon}|\theta')>\frac{(1-\epsilon)\rho f^Q_{\theta}(s)}{\epsilon f^Q_{\theta'}(s)}.
\end{align*}
For sufficiently small $\epsilon$, we have $\frac{(1-\epsilon)\rho f^Q_{\theta}(s)}{\epsilon f^Q_{\theta'}(s)}>1$, which is a contradiction.
\end{proof}

\begin{lem}
There exists an equilibrium $\phi^*$ such that given any realized belief $r\in(0,1)$ about state $1$ for an agent observing $B_k$, for any $\hat{r}\in(0,r)$ ($\hat{r}\in(r,1)$), $N(r,\hat{r})\in\mathbb{N}$ exists such that a realized belief that is less than $\hat{r}$ (higher than $\hat{r}$) can be induced by having complete observations of additional $N(r,\hat{r})$ communities, each with unanimous action $0$ ($1$).
\end{lem}

\begin{proof}
Without loss of generality, assume that $\hat{r}\in(0,r)$. We know that there is a private signal $s$ and an action sequence $h$ from $B_k$ such that
\begin{align*}
r=\frac{\mathcal{P}_{\phi^*}(h|1)f^Q_1(s)}{\mathcal{P}_{\phi^*}(h|1)f^Q_1(s)+\mathcal{P}_{\phi^*}(h|0)f^Q_0(s)}.
\end{align*}
Let $a^{k+1}=0$ denote the event that action $0$ is taken by every agent in the $(k+1)$th community. The new belief would then be
\begin{align*}
r_1=\frac{\mathcal{P}_{\phi^*}(h|1)f^Q_1(s)\times \mathcal{P}_{\phi^*}(a^{k+1}=0|h,1)}{\mathcal{P}_{\sigma^*(1)}(h|1)f^Q_1(s)\times \mathcal{P}_{\phi^*}(a^{k+1}=0|h,1)+\mathcal{P}_{\phi^*}(h|0)f^Q_0(s)\times \mathcal{P}_{\phi^*}(a^{k+1}=0|h,0)}.
\end{align*}

I now explicitly describe an equilibrium $\phi^*$ that will prove this result. Consider the following strategy profile for agents in an arbitrary community $N^t$:

\begin{itemize}
\item{1.} If $Q^t<\hat{Q}$, each agent takes the action that matches the state with higher probability according to signal and observation.
\item{2.} If $Q^t\geq\hat{Q}$:
\begin{itemize}
\item Fix some $\epsilon>0$. Let $s_1(\epsilon),s_0(\epsilon)$ be such that $F^Q_1(s_1(\epsilon))-F^Q_1(s_0(\epsilon))=F^Q_0(s_1(\epsilon))-F^Q_0(s_0(\epsilon))=1-2\epsilon$. Each agent takes action $1$ if $s^t\geq s_1(\epsilon)$ and action $0$ if $s^t\leq s_0(\epsilon)$.
\item Otherwise, each agent takes the action that matches the state with higher probability according to observation only.
\end{itemize}
\end{itemize}

By Lemma 1, when $Q$ is sufficiently large, both (1) and (2) constitute mutual best responses given $I^t$. Hence, the above strategy profile is an equilibrium. We then have
\begin{align*}
\mathcal{P}_{\phi^*}(a^{k+1}=0|h,1)=&F^Q_1(s_0(\epsilon))+(F^Q_1(s_1(\epsilon))-F^Q_1(s_0(\epsilon)))\\
&\mathcal{P}_{\phi^*}(a^{k+1}=0|s^{k+1}\in(s_0(\epsilon),s_1(\epsilon)),h,1)\\
\mathcal{P}_{\phi^*}(a^{k+1}=0|h,0)=&F^Q_0(s_0(\epsilon))+(F^Q_0(s_1(\epsilon))-F^Q_0(s_0(\epsilon)))\\
&\mathcal{P}_{\phi^*}(a^{k+1}=0|s^{k+1}\in(s_0(\epsilon),s_1(\epsilon)),h,0)
\end{align*}

By construction of $\phi^*$, $F^Q_1(s_0(\epsilon))<F^Q_0(s_0(\epsilon))$. Based on the assumption of complete observations, $h$ consists of every action that agent $k+1$ may have observed, which implies that $\mathcal{P}_{\phi^*}(a^{k+1}=0|s^{k+1}\in(s_0(\epsilon),s_1(\epsilon)),h,1)=\mathcal{P}_{\phi^*}(a^{k+1}=0|s^{k+1}\in(s_0(\epsilon),s_1(\epsilon)),h,0)$. Hence, we know that the ratio $\frac{\mathcal{P}_{\phi^*}(a^{k+1}=0|h,1)}{\mathcal{P}_{\phi^*}(a^{k+1}=0|h,0)}$ has a $<1$ upper bound that is independent of $h$. Let $y$ denote this bound, and we have
\begin{align*}
&\frac{r_1}{r}=\frac{1+\frac{\mathcal{P}_{\phi^*}(h|0)f^Q_0(s)}{\mathcal{P}_{\phi^*}(h|1)f^Q_1(s)}}{1+\frac{\mathcal{P}_{\phi^*}(h|0)f^Q_0(s)}{\mathcal{P}_{\phi^*}(h|1)f^Q_1(s)}\frac{\mathcal{P}_{\phi^*}(a^{k+1}=0|h,1)}{\mathcal{P}_{\phi^*}(a^{k+1}=0|h,0)}}\\
=&\frac{1}{r+(1-r)\frac{\mathcal{P}_{\phi^*}(a^{k+1}=0|h,0)}{\mathcal{P}_{\phi^*}(a^{k+1}=0|h,1)}}<\frac{1}{r+(1-r)\frac{1}{y}}<1.
\end{align*}

Note that the expression on the right-hand side above is increasing in $r$. Let $r_m$ denote the belief induced by $h\cup\{a^{k+1},\cdots, a^{k+m}\}$ where $a^{k+1}=\cdots=a^{k+m}=0$. We have
\begin{align*}
r_m=r\times\frac{r_1}{r}\times\cdots\times\frac{r_m}{r_{m-1}}<r\times\frac{1}{r+(1-r)\frac{1}{y}}\times\cdots\times\frac{1}{r_m+(1-r_m)\frac{1}{y}}.
\end{align*}

Because $r,r_1,\cdots,r_m$ are decreasing, we have $r_m<r\times(\frac{1}{r+(1-r)\frac{1}{y}})^m$. Hence, we can find the desired $N(r,\hat{r})$ for any $\hat{r}\in(0,r)$ such that a realized belief that is less than $\hat{r}$ can be induced by $s$ and $h\cup\{a^{k+1},\cdots,a^{k+{N(r,\hat{r})}}\}$, where $a^{k+1}=\cdots=a^{k+{N(r,\hat{r})}}=0$.
\end{proof}

\begin{lem}
Consider the $\phi^*$ constructed above. Let $\hat{a}$ be the action that matches the state with higher probability given $s$ and every action in $B_k$, and let $\mathcal{P}_{\phi^*}^{B_k}(\hat{a}\neq \theta|s)$ denote the probability that $\hat{a}$ does not match the state. We have $\lim_{k\rightarrow\infty}\mathcal{P}_{\phi^*}^{B_k}(\hat{a}\neq \theta|s)=0$.
\end{lem}
%The meaning of the following sentence was unclear. Please confirm that your intended meaning has been maintained in these edits.
\begin{proof}
Suppose the opposite: noting that $\mathcal{P}_{\phi^*}^{B_k}(\hat{a}\neq \theta|s)$ must be weakly decreasing in $k$, it follows that $\lim_{k\rightarrow\infty}\mathcal{P}_{\phi^*}^{B_k}(\hat{a}\neq \theta|s)>0$. Let $\rho>0$ denote this limit. From Lemma 2, we know that for any $\alpha>0$ and for either true state $\theta=0,1$, $z\in[\frac{1}{2},1)$ exists such that $M\in\mathbb{N}$ exists such that $\max\{\mathcal{P}_{\phi^*}(R_{\phi^*}^{B_{k}}>z|0,s),\mathcal{P}_{\phi^*}(1-R_{\phi^*}^{B_{k}}>z|1,s)\}<\alpha$ for any $k>M$. If $\alpha=\frac{1}{2}\rho$, then we have $\max\{\mathcal{P}_{\phi^*}(R_{\phi^*}^{B_k}>z|0,s),\mathcal{P}_{\phi^*}(1-R_{\phi^*}^{B_k}>z|1,s)\}<\frac{1}{2}\rho$ for any $k>M$. Then, for any $\delta>0$, we can find a sufficiently large $k$ such that for any $k'\geq k$, (1) $\mathcal{P}_{\phi^*}^{B_{k'}}(\hat{a}\neq \theta|s)\in(\rho,\rho+\delta)$ and (2) $\max\{\mathcal{P}_{\phi^*}(R_{\phi^*}^{B_{k'}}>z|0,s),\mathcal{P}_{\phi^*}(1-R_{\phi^*}^{B_{k'}}>z|1,s)\}<\frac{1}{2}\rho$. Hence, we have
\begin{align*}
&\frac{f^Q_0(s)}{f^Q_0(s)+f^Q_1(s)}\mathcal{P}_{\phi^*}(R_{\phi^*}^{B_{k'}}\in[\frac{1}{2},z]|0,s)+\frac{f^Q_1(s)}{f^Q_0(s)+f^Q_1(s)}\mathcal{P}_{\phi^*}(1-R_{\phi^*}^{B_{k'}}\in[\frac{1}{2},z]|1,s)\\
=&\mathcal{P}_{\phi^*}^{B_{k'}}(\hat{a}\neq \theta|s)-\frac{f^Q_0(s)}{f^Q_0(s)+f^Q_1(s)}\mathcal{P}_{\phi^*}(R_{\phi^*}^{B_{k'}}>z|0,s)\\
&-\frac{f^Q_1(s)}{f^Q_0(s)+f_1(s)}\mathcal{P}_{\phi^*}(1-R_{\phi^*}^{B_{k'}}>z|1,s)>\frac{1}{2}\rho.
\end{align*}

By Lemma 3, for any $\pi>0$, $N(\pi)=\max\{N(z,\frac{1}{2+\pi}),N(1-z,1-\frac{1}{2+\pi})\}\in\mathbb{N}$ exists such that whenever $\theta=0$ and $R_{\phi^*}^{B_k}\in[\frac{1}{2},z]$ or $\theta=1$ and $1-R_{\phi^*}^{B_k}\in[\frac{1}{2},z]$, additional $N(\pi)$ observations can reverse an incorrect decision. Consider the following (sub-optimal) updating method for a rational agent who observes $B_{k'}=B_{k+N(\pi)}$: her action changes from $1$ to $0$ if and only if $R_{\phi^*}^{B_k}\in[\frac{1}{2},z]$, and $a^{k+1}=\cdots=a^{k+N(\pi)}=0$; her action changes from $0$ to $1$ if and only if $1-R_{\phi^*}^{B_k}\in[\frac{1}{2},z]$, and $a^{k+1}=\cdots=a^{k+N(\pi)}=1$. Let $h$ denote a decision sequence from $B_k$ that, together with $s$, induces such a posterior belief in the former case, and let $h'$ denote a decision sequence from $B_k$ that, together with $s$, induces such a posterior belief in the latter case. Let $H$ and $H'$ respectively denote the sets of these decision sequences. We have
\begin{align*}
&\mathcal{P}_{\phi^*}^{B_k}(\hat{a}\neq \theta|s)-\mathcal{P}_{\phi^*}^{B_{k'}}(\hat{a}\neq \theta|s)\\
\geq&\sum_{h\in H}(\frac{f^Q_0(s)}{f^Q_0(s)+f^Q_1(s)}\mathcal{P}_{\phi^*}(h,a^{k+1}=\cdots=a^{k+N(\pi)}=0|0)\\
&-\frac{f^Q_1(s)}{f^Q_0(s)+f^Q_1(s)}\mathcal{P}_{\phi^*}(h,a^{k+1}=\cdots=a^{k+N(\pi)}=0|1))\\
&+\sum_{h'\in H'}(\frac{f^Q_1(s)}{f^Q_0(s)+f^Q_1(s)}\mathcal{P}_{\phi^*}(h',a^{k+1}=\cdots=a^{k+N(\pi)}=1|1)\\
&-\frac{f^Q_0(s)}{f^Q_0(s)+f^Q_1(s)}\mathcal{P}_{\phi^*}(h',a^{k+1}=\cdots=a^{k+N(\pi)}=1|0)).
\end{align*}
From the proof of Lemma 3, we know that for every $h$,
\begin{align*}
&\frac{\mathcal{P}_{\phi^*}(h,a^{k+1}=\cdots=a^{k+N(\pi)}=0|0)f^Q_0(s)}{\mathcal{P}_{\phi^*}(h,a^{k+1}=\cdots=a^{k+N(\pi)}=0|0)f^Q_0(s)+\mathcal{P}_{\phi^*}(h,a^{k+1}=\cdots=a^{k+N(\pi)}=0|1)f^Q_1(s)}\\
\geq&\frac{1+\pi}{2+\pi},
\end{align*}
which implies that
\begin{align*}
&\mathcal{P}_{\phi^*}(h,a^{k+1}=\cdots=a^{k+N(\pi)}=0|0)f^Q_0(s)-\mathcal{P}_{\phi^*}(h,a^{k+1}=\cdots=a^{k+N(\pi)}=0|1)f^Q_1(s)\\
\geq&\pi f^Q_1(s)\mathcal{P}_{\phi^*}(h,a^{k+1}=\cdots=a^{k+N(\pi)}=0|1).
\end{align*}
From the proof of Lemma 3, we know that the quantities $\mathcal{P}_{\phi^*}(a^{k+1}=0|h,1)$ and $\mathcal{P}_{\phi^*}(a^{k+1}=1|h,0)$ have a $>0$ lower bound that is independent of $h$. We denote this bound as $w$, and the above inequality can be written as
\begin{align*}
&\mathcal{P}_{\phi^*}(h,a^{k+1}=\cdots=a^{k+N(\pi)}=0|0)f^Q_0(s)-\mathcal{P}_{\phi^*}(h,a^{k+1}=\cdots=a^{k+N(\pi)}=0|1)f^Q_1(s)\\
\geq&\pi f^Q_1(s)w^{N(\pi)}\mathcal{P}_{\phi^*}(h|1).
\end{align*}
Based on the definition of $h$, we have
\begin{align*}
\frac{1}{2}\leq\frac{\mathcal{P}_{\phi^*}(h|1)f^Q_1(s)}{\mathcal{P}_{\phi^*}(h|1)f^Q_1(s)+\mathcal{P}_{\phi^*}(h|0)f^Q_0(s)}\leq z,
\end{align*}
which implies that
\begin{align*}
\mathcal{P}_{\phi^*}(h|1)f^Q_1(s)\geq \mathcal{P}_{\phi^*}(h|0)f^Q_0(s).
\end{align*}
Similarly, we have
\begin{align*}
&\mathcal{P}_{\phi^*}(h',a^{k+1}=\cdots=a^{k+N(\pi)}=1|1)f^Q_1(s)-\mathcal{P}_{\phi^*}(h',a^{k+1}=\cdots=a^{k+N(\pi)}=1|0)f^Q_0(s)\\
\geq&\pi f^Q_0(s)w^{N(\pi)}\mathcal{P}_{\phi^*}(h'|0),
\end{align*}
and
\begin{align*}
\mathcal{P}_{\phi^*}(h'|0)f^Q_0(s)\geq \mathcal{P}_{\phi^*}(h'|1)f^Q_1(s).
\end{align*}
From the previous construction, we know that
\begin{align*}
&\frac{f^Q_0(s)}{f^Q_0(s)+f^Q_1(s)}\sum_{h\in H} \mathcal{P}_{\phi^*}(h|0)+\frac{f^Q_1(s)}{f^Q_0(s)+f^Q_1(s)}\sum_{h'\in H'} \mathcal{P}_{\phi^*}(h'|1)\\
=&\frac{f^Q_0(s)}{f^Q_0(s)+f^Q_1(s)}\mathcal{P}_{\phi^*}(R_{\phi^*}^{B_{k'}}\in[\frac{1}{2},z]|0,s)+\frac{f^Q_1(s)}{f^Q_0(s)+f^Q_1(s)}\mathcal{P}_{\phi^*}(1-R_{\phi^*}^{B_{k'}}\in[\frac{1}{2},z]|1,s)\\
>&\frac{1}{2}\rho.
\end{align*}
Combining the previous inequalities, we have
\begin{align*}
&\mathcal{P}_{\phi^*}^{B_k}(\hat{a}\neq \theta|s)-\mathcal{P}_{\phi^*}^{B_{k'}}(\hat{a}\neq \theta|s)\\
>&\pi w^{N(\pi)}\frac{1}{2}\rho.
\end{align*}
From the previous construction, we also know that
\begin{align*}
\mathcal{P}_{\phi^*}^{B_k}(\hat{a}\neq \theta|s)-\mathcal{P}_{\phi^*}^{B_{k'}}(\hat{a}\neq \theta|s)<\delta.
\end{align*}
Clearly, for some given $\pi>0$, a sufficiently small $\delta$ exists such that $\pi w^{N(\pi)}\frac{1}{2}\rho>\delta$, which implies a contradiction.
\end{proof}

Lemma 4 implies that in the equilibrium $\phi^*$ constructed in Lemma 3, truth-telling observation occurs. We can then compute the probability of taking the state-matching action: $\lim_{t\rightarrow\infty}\mathcal{P}_{\phi^*}(a_n^t=\theta)=1-\frac{1}{2}F^Q_1(s_0(\epsilon))-\frac{1}{2}(1-F^Q_0(s_1(\epsilon))$. From the characterization of $s_0(\epsilon)$ and $s_1(\epsilon)$, we know that $F^Q_1(s_0(\epsilon))+1-F^Q_0(s_1(\epsilon))<F^Q_1(s_0(\epsilon))+1-F^Q_1(s_1(\epsilon))=2\epsilon$, which implies that $\lim_{t\rightarrow\infty}\mathcal{P}_{\phi^*}(a_n^t=\theta)>1-\epsilon$.
\end{proof}

\begin{proof}[Proof of Proposition 2]

Consider an infinite subset of communities $\{N^{t_1},N^{t_2},\cdots\}$ such that $N^{t_{n+1}}$'s observation includes $N^{t_{n}}$. By the assumption of expanding observations, we know that each community belongs to at least one such subset. Then, it suffices to show that asymptotic learning occurs in each of these subsets.

Consider the following equilibrium $\phi^*$: for every $I^t$, every agent in $N^t$ takes the action that matches the state with higher probability according to both signal and observation. Because actions are unanimous in equilibrium, I use $a^t$ to denote the equilibrium action in community $N^t$. Suppose that the result does not hold, i.e., there exists $\{N^{t_1},N^{t_2},\cdots\}$ such that $\lim_{n\rightarrow}\mathcal{P}_{\phi^*}(a^{t_n}=\theta)=P<1$.

For an arbitrary $\epsilon>0$, find $n$ such that $\mathcal{P}_{\phi^*}(a^{t_n}=\theta)\in(P-\epsilon,P)$. It follows that there must be some community size $Q'$ in support of $G$, such that either $\mathcal{P}_{\phi^*}(a^{t_n}=1|\theta=1,Q^{t_n}=Q')<P$ or $\mathcal{P}_{\phi^*}(a^{t_n}=0|\theta=0,Q^{t_n}=Q')<P$. Without loss of generality, assume that the second inequality holds. Based on the assumption of unbounded private beliefs, there exists $Q$ such that $\lim_{s\rightarrow 1}\frac{f^Q_1(s)}{f^Q_0(s)+f^Q_1(s)}=1$ and $\lim_{s\rightarrow -1}\frac{f^Q_1(s)}{f^Q_0(s)+f^Q_1(s)}=0$. Consider community $N^{t_{n+1}}$ with the realized community size $Q$ and a sufficiently small private signal $s^{t_{n+1}}$. If the agents in $N^{t_{n+1}}$ take their action according to signals only, each receives the expected payoff
\begin{align*}
\frac{f_0^Q(s^{t_{n+1}})}{f_0^Q(s^{t_{n+1}})+f_1^Q(s^{t_{n+1}})}u(1,1,Q)+\frac{f_1^Q(s^{t_{n+1}})}{f_0^Q(s^{t_{n+1}})+f_1^Q(s^{t_{n+1}})}u(1,0,Q).
\end{align*}
If they simply follow the action in community $N^{t_{n}}$, each receives the expected payoff
\begin{align*}
&\frac{f_0^Q(s^{t_{n+1}})}{f_0^Q(s^{t_{n+1}})+f_1^Q(s^{t_{n+1}})}(\mathcal{P}_{\phi^*}(a^{t_n}=0|\theta=0,Q^{t_n}=Q')u(1,1,Q)\\
&+\mathcal{P}_{\phi^*}(a^{t_n}=1|\theta=0,Q^{t_n}=Q')u(1,0,Q))\\
+&\frac{f_1^Q(s^{t_{n+1}})}{f_0^Q(s^{t_{n+1}})+f_1^Q(s^{t_{n+1}})}(\mathcal{P}_{\phi^*}(a^{t_n}=1|\theta=1,Q^{t_n}=Q')u(1,1,Q)\\
&+\mathcal{P}_{\phi^*}(a^{t_n}=0|\theta=1,Q^{t_n}=Q')u(1,0,Q)).
\end{align*}
The difference between the two payoffs is bounded below by
\begin{align*}
(\frac{f_0^Q(s^{t_{n+1}})}{f_0^Q(s^{t_{n+1}})+f_1^Q(s^{t_{n+1}})}(1-P)-\frac{f_1^Q(s^{t_{n+1}})}{f_0^Q(s^{t_{n+1}})+f_1^Q(s^{t_{n+1}})})(u(1,1,Q)-u(1,0,Q)).
\end{align*}

Because private beliefs are bounded, there exists $s'>-1$ such that this difference is positive when $s^{t_{n+1}}\in(-1,s')$. It follows that there exists $\Delta>0$ such that $\mathcal{P}_{\phi^*}(a^{t_{n+1}}=\theta|Q^{t_{n+1}}=Q)>\mathcal{P}_{\phi^*}(a^{t_{n}}=\theta|Q^{t_{n}}=Q')+\Delta$. Hence, we have $\mathcal{P}_{\phi^*}(a^{t_{n+1}}=\theta)>\mathcal{P}_{\phi^*}(a^{t_{n}}=\theta)+\Delta G(Q')G(Q)>P-\epsilon+\Delta G(Q')G(Q)$. Because $\epsilon$ can be taken arbitrarily, when $\epsilon$ is sufficiently small, we have $\mathcal{P}_{\phi^*}(a^{t_{n+1}}=\theta)>P$, which is a contradiction.
\end{proof}

\begin{proof}[Proof of Theorem 3]

Consider the following strategy profile $\sigma$ for agents in an arbitrary community $N^t$:

\begin{itemize}
\item{1.} Given any $s^t$, agent $1$ observes $a^{t-1}_1$, and no other agent makes any observation.
\item{2.} If $h^t=\{a^{t-1}_1\}$, then each agent in $N^t$ takes the action that matches the state with higher probability according to $I^t$. Otherwise, each agent takes the opposite action (the action that matches the state with lower probability).
\end{itemize}

By Lemma 1, the action profile given $I^t$ specified above constitutes mutual best responses when $Q^t$ is sufficiently large. If $h^t\neq\{a^{t-1}_1\}$, the payoff before cost for each agent in $N^t$ is bounded above by $\frac{1}{2}(u(1,1,Q^t)+u(1,0,Q^t))$; if $h^t=\{a^{t-1}_1\}$, the payoff before cost is bounded below by
\begin{align*}
\frac{1}{2}((F_0^{Q^t}(\hat{s})+(1-F_1^{Q^t}(\hat{s})))u(1,1,Q^t)+(F_1^{Q^t}(\hat{s})+(1-F_0^{Q^t}(\hat{s})))u(1,0,Q^t)),
\end{align*}
where $\hat{s}$ is characterized by $f^{Q^t}_0(\hat{s})=f^{Q^t}_1(\hat{s})$. When $c$ is small, the difference between the two payoffs exceeds $c$ for sufficiently large $Q^t$. Hence, it is optimal to follow the observation decision in (1) above given that every other agent follows $\sigma$, which means that $\sigma$ is an equilibrium.

Note that in $\sigma$, starting from $t=2$, agents in $Q^t$ always observe $a^{t-1}_1$ regardless of $s^t$. Thus, we can apply Proposition 2 to obtain asymptotic learning in $\sigma$.
\end{proof}

\begin{proof}[Proof of Theorem 4]
Consider the following strategy profile $\sigma$ for agents in an arbitrary community $N^t$:
\begin{itemize}
\item{1.} Given any $s^t$, agent $1$ observes the neighborhood $B^t$ of size $K(t)$ that maximizes $\mathcal{P}_{\sigma}(\hat{a}^t=\theta|B^t)$, and no other agent makes any observation.
\item{2.} If $h^t=\{a_m:m\in B^t\}$: Fix some $\epsilon>0$. Let $s_1(\epsilon),s_0(\epsilon)$ be such that $F^{Q^t}_1(s_1(\epsilon))-F^{Q^t}_1(s_0(\epsilon))=F^{Q^t}_0(s_1(\epsilon))-F^{Q^t}_0(s_0(\epsilon))=1-2\epsilon$. An agent takes action $1$ if $s^t\geq s_1(\epsilon)$ and action $0$ if $s^t\leq s_0(\epsilon)$. Otherwise, an agent takes the action that matches the state with higher probability according to observation only.
\item{3.} If $h^t\neq\{a_m:m\in B^t\}$, each agent takes the action that matches the state with lower probability according to $I^t$.
\end{itemize}

From the proofs of Lemma 1 and Theorem 3, $\sigma$ is an equilibrium. We can then apply Theorem 1 to prove the result.
\end{proof}

\begin{proof}[Proof of Proposition 3]
Consider any community $N^t$ of size $Q$, $N^{'t}\subset N^t$ of size $Q'$ and any $I^t$. Let $a^t(I^t)$ be the truth-seeking action profile and $a^{'t}(I^t)$ be an arbitrary action profile with unanimous action. Without loss of generality, assume that $a_n^t(I^t)=1$ and $a_n^{'t}(I^t)=0$. Let $P$ denote the probability that $\theta=1$ given $I^t$. The definition of the truth-seeking action profile implies that $P\geq\frac{1}{2}$. Then, we have
\begin{align*}
&v_n^t(a^t(I^t),I^t)-v_n^t((\{a_i^{'t}(I^t)\}_{i\in N^{'t}},\{a_j^{t}(I^t)\}_{i\notin N^{'t}}),I^t)\\
=&Pu(1,1,Q)+(1-P)u(1,0,Q)-(Pu(1,0,Q')+(1-P)u(1,1,Q'))\\
&v_n^t(a^{'t}(I^t),I^t)-v_n^t((\{a_i^{t}(I^t)\}_{i\in N^{'t}},\{a_j^{'t}(I^t)\}_{i\notin N^{'t}}),I^t)\\
=&Pu(1,0,Q)+(1-P)u(1,1,Q)-(Pu(1,1,Q')+(1-P)u(1,0,Q')).
\end{align*}
It follows that
\begin{align*}
&v_n^t(a^t(I^t),I^t)-v_n^t((\{a_i^{'t}(I^t)\}_{i\in N^{'t}},\{a_j^{t}(I^t)\}_{i\notin N^{'t}}),I^t)\\
&-(v_n^t(a^{'t}(I^t),I^t)-v_n^t((\{a_i^{t}(I^t)\}_{i\in N^{'t}},\{a_j^{'t}(I^t)\}_{i\notin N^{'t}}),I^t))\\
=&(2P-1)(u(1,1,Q)-u(1,0,Q'))+(1-2P)(u(1,0,Q)-u(1,1,Q'))\\
=&(2P-1)(u(1,1,Q)-u(1,0,Q')-(u(1,0,Q)-u(1,1,Q')))\\
=&(2P-1)(u(1,1,Q)-u(1,0,Q)+u(1,1,Q')-u(1,0,Q'))\geq 0.
\end{align*}
Hence, the inequality is proven.
\end{proof}

\begin{proof}[Proof of Proposition 4]
%Please ensure that the intended meaning has been maintained in third sentence below.
From Theorem 2, we know that truth-telling observation occurs in every $\sigma^*_T$. Suppose that $Q^t=Q$. From the characterization of $s^*(Q)$, we know that for any $s^t\in(-s^*(Q),s^*(Q))$, agents in $N^t$ prefer paying $c$ to know the true state over paying nothing and acting according to $s^t$. It then follows that when $t$ is sufficiently large, whenever $s^t\in(-s^*,s^*)$, the equilibrium observation in $\sigma^*_T$ is non-empty; otherwise, given the truth-seeking action profile, any agent can be better off by paying $c$ and observing a neighborhood of size $K(t)$. Therefore, at the limit, an agent takes the correct action if and only if her signal lies in $[-s^*(Q),s^*(Q)]$ and follows her signal otherwise. The probability of her action matching the state, $\mathcal{P}_{\sigma^*_T}(a_n^t=\theta|Q^t=Q)$, is then equal to $F_0(s^*(Q))$. Finally, by taking expectation over the probability distribution $G$, we obtain $\mathcal{P}_{\sigma^*_T}(a_n^t=\theta)=\mathbb{E}_G[F_0(s^*(Q))]$.

\end{proof}

\begin{proof}[Proof of Corollary 1]

Note that $F_0(s^*(Q))$ is increasing in $Q$. The result follows by applying the relevant property of first-order stochastic dominance.

\end{proof}

\begin{proof}[Proof of Proposition 5]
Suppose that there exists a symmetric equilibrium $\phi^*$ with asymptotic learning. Since agents in the same community are using identical strategies by assumption, I use $s^t$ to denote the private signal of an arbitrary agent in $N^t$.

Without loss of generality, suppose that the true state is $1$. Fix an arbitrary positive number $\epsilon_1$. Since asymptotic learning occurs, we can find $\hat{t}$ such that: if the true state is $0$, for every $t\geq\hat{t}$ and every agent in $N^{t}$, the agent takes action $0$ with a probability of at least $1-\epsilon_1$. Hence the probability that every agent in $N^{\hat{t}}$ takes action $0$ is at least $(1-\epsilon_1)^{Q^{\hat{t}}}$. By asymptotic learning and continuity of the signal distributions, for every positive number $\epsilon_2$ we can find a sufficiently small $\epsilon_1$ and the corresponding $\hat{t}$ such that an arbitrary agent in $N^{\hat{t}+1}$ takes action $0$ with a probability of at least $1-\epsilon_2$ under either state, given that each agent in $N^{\hat{t}}$ has taken action $0$. Note that since $\hat{t}$ is finite, the probability that each agent in $N^{\hat{t}}$ takes action $0$ under the true state $1$ is still positive and bounded away from $0$.

Now consider the optimal action of an agent in $N^{\hat{t}+1}$. Signals are independent, so with probability of at least $(1-\epsilon_2)^{Q^{\hat{t}+1}-1}$ all the other agents in $N^{\hat{t}+1}$ will choose action $0$. If the agent chooses action $0$, her payoff is bounded below by $(1-\epsilon_2)^{Q^{\hat{t}+1}-1}u(1,0,Q^{\hat{t}+1})$; if she chooses action $1$, her payoff is bounded above by $\max\{(1-\epsilon_2)^{Q^{\hat{t}+1}-1}u(1,1,1)+(1-(1-\epsilon_2)^{Q^{\hat{t}+1}-1})u(1,0,Q^{\hat{t}+1}),(1-\epsilon_2)^{Q^{\hat{t}+1}-1}u(1,0,1)+(1-(1-\epsilon_2)^{Q^{\hat{t}+1}-1})u(1,1,Q^{\hat{t}+1})\}$. When $Q^{\hat{t}+1}$ is sufficiently large, we can find sufficiently small $\epsilon_2$ such that the first bound is higher than the second bound. Hence, the agent will choose action $0$ regardless of her private signal. This argument then works for every $t'\geq \hat{t}+1$, which means that herding occurs. This is a contradiction to asymptotic learning.
\end{proof}

\begin{proof}[Proof of Proposition 6]

In any equilibrium, the number of agents choosing action $1$, denoted $Q^t_1$, must satisfy
\begin{align*}
&Pu(1,1,Q^t_1)+(1-P)u(1,0,Q^t_1)\geq(1-P)u(1,1,Q^t-Q^t_1+1)+Pu(1,0,Q^t-Q^t_1+1)\\
&(1-P)u(1,1,Q^t-Q^t_1)+Pu(1,0,Q^t-Q^t_1)\geq Pu(1,1,Q^t_1+1)+(1-P)u(1,0,Q^t_1+1).
\end{align*}

From the second inequality, we have
\begin{align*}
u(1,1,Q^t-Q^t_1)-u(1,0,Q^t_1+1)\geq\frac{P}{1-P}(u(1,1,Q^t_1+1)-u(1,0,Q^t-Q^t_1))
\end{align*}
From the assumption that $P>1$, we know that either (1) $u(1,1,Q^t-Q^t_1)-u(1,0,Q^t_1+1)$ is positive and $u(1,1,Q^t-Q^t_1)-u(1,0,Q^t_1+1)\geq u(1,1,Q^t_1+1)-u(1,0,Q^t-Q^t_1)$ or (2) both $u(1,1,Q^t-Q^t_1)-u(1,0,Q^t_1+1)$ and $u(1,1,Q^t_1+1)-u(1,0,Q^t-Q^t_1)$ are non-positive.

Assume that (1) holds. Note that $u(1,1,Q^t-Q^t_1)-u(1,1,Q^t_1+1)$ and $u(1,1,Q^t-Q^t_1)-u(1,1,Q^t_1+1)$ have the same sign. Then, from $u(1,1,Q^t-Q^t_1)-u(1,0,Q^t_1+1)\geq u(1,1,Q^t_1+1)-u(1,0,Q^t-Q^t_1)$, we have $Q^t-Q^t_1\leq Q^t_1+1$, which implies that $Q^t_1\geq \frac{1}{2}Q^t$.

Assume that (2) holds. Then, from the two non-positive expressions, it follows that $Q^t-Q^t_1>Q^t_1+1$ and $Q^t_1+1>Q^t-Q^t_1$, which is a contradiction. Hence, we can conclude that the only possible case is (1), and thus, $Q^t_1\geq \frac{1}{2}Q^t$.
\end{proof}

\newpage
\bibliographystyle{acm}
\bibliography{reference}

\end{document}